\documentclass[
  aps,
  physrev,
  reprint,
  superscriptaddress,
  nofootinbib,
  longbibliography,
  floatfix
]{revtex4-2}

\usepackage[T1]{fontenc}
\usepackage[utf8]{inputenc}
\usepackage{microtype}
\usepackage{amsmath,amssymb,amsfonts,bm,mathtools,amsthm}
\usepackage{braket}
\usepackage{graphicx}
\graphicspath{{figures/}}
\usepackage{xcolor}
\usepackage{orcidlink}
\usepackage{booktabs}
\usepackage{silence}
\usepackage{hyperref}

\hypersetup{
  colorlinks=true,
  linkcolor=blue!50!black,
  citecolor=blue!50!black,
  urlcolor=blue!50!black,
  pdftitle={Collective thermalization, work reliability, and resource bounds in a population-inverted Dicke Otto engine},
  pdfauthor={Cler T. Garcez, Gabriella G. Damas, Norton G. de Almeida, and G. D. de Moraes Neto},
  pdfsubject={Collective quantum Otto engines, population inversion, resource accounting, and work fluctuations}
}

\newtheorem{theorem}{Theorem}[section]
\newtheorem{proposition}[theorem]{Proposition}

\theoremstyle{definition}

\newcommand{\avg}[1]{\left\langle #1\right\rangle}

\newcommand{\Tr}{\operatorname{Tr}}
\newcommand{\diag}{\operatorname{diag}}

\newcommand{\ind}{\mathrm{ind}}

\newcommand{\gross}{\mathrm{gross}}

\newcommand{\dd}{\mathrm d}

\newcommand{\id}{\mathbb I}

\begin{document}
\raggedbottom

\title{Collective thermalization, work reliability, and resource bounds in a population-inverted Dicke Otto engine}

\author{Cler T. Garcez~\orcidlink{0009-0003-0999-3379}}
%\email{clergarcez@discente.ufg.br}
\affiliation{Instituto de F\'isica, Universidade Federal de Goi\'as, 74.001-970, Goi\^ania--GO, Brazil}

\author{Gabriella G. Damas~\orcidlink{0000-0003-3376-9281}}
\affiliation{Instituto de F\'isica, Universidade Federal de Goi\'as, 74.001-970, Goi\^ania--GO, Brazil}
\affiliation{Department of Physics, Zhejiang Normal University, Jinhua 321004, China}

\author{Norton G. de Almeida~\orcidlink{0000-0001-8517-6774}}
\email{norton@ufg.br}
\affiliation{Instituto de F\'isica, Universidade Federal de Goi\'as, 74.001-970, Goi\^ania--GO, Brazil}

\author{G. D. de Moraes Neto~\orcidlink{0000-0003-4273-8380}}
\email{gdmneto@gmail.com}
\affiliation{Department of Fundamental Sciences, Hainan Bielefeld University of Applied Sciences, Danzhou, Hainan 578101, China}

\date{September 2026}

\begin{abstract}
Population inversion and collective relaxation can both enhance the performance of a quantum heat
engine, but they do so for fundamentally different reasons. We disentangle these effects in a
quantum Otto engine whose working medium is a symmetric collective spin of $N$ two-level
constituents. For commuting work strokes and collective reservoir coupling, the dynamics reduces
exactly to a finite birth--death process on the Dicke ladder, providing a unified description of
stationary operation, work fluctuations, and finite-time relaxation. In the complete-reset limit,
where each reservoir contact fully relaxes the working medium, the passive work per cycle saturates
as the system grows, whereas population inversion makes the extracted work increase linearly with
the number of constituents. The work reliability exhibits the same linear improvement, exceeding
the square-root scaling obtained from $N$ independent engines. The origin of this enhancement is
geometric: passive hot and cold states remain localized near the same end of the Dicke ladder,
whereas inversion places them near opposite ends, producing a macroscopic polarization
displacement. Collective coupling plays a different role by accelerating the finite-time dynamics
through the Dicke transition rates. For matched passive and inverted hot states, the additional
gross work is exactly the Otto efficiency times the ergotropy of the inverted state. Once the
corresponding reversible excess state-formation cost is charged, the incremental advantage over the
passive engine becomes negative, showing that the enhanced gross output converts a pre-existing
active-state resource rather than creating a free thermodynamic gain. For finite contacts, exact
trajectory statistics and tilted cycle maps further reveal intra-cycle and inter-cycle
correlations and quantify how collective kinetics, fluctuations, and resource accounting remain
linked away from complete reset.
\end{abstract}

\maketitle

\section{Introduction}
\label{sec:intro}

Quantum Otto cycles provide a transparent setting in which changes of the working-medium Hamiltonian are separated from energy exchange with reservoirs, allowing work and heat to be identified directly at the level of a microscopic quantum system~\cite{kieu2004second,quan2007quantum,rezek2006irreversible,alicki2015nonequilibrium}. At finite cycle times, however, the mean extracted work is only part of the thermodynamic description. Incomplete contacts make energy exchange stochastic, correlate different cycle corners, and can preserve memory from one cycle to the next~\cite{talkner2007fluctuation,esposito2009nonequilibrium,campisi2011colloquium,manikandan2019efficiency,denzler2020efficiency,denzler2024nonequilibrium,saryal2021bounds}. These fluctuations are experimentally accessible: full-counting statistics of quantized heat exchange and output fluctuations have been resolved in a quasi-spin Otto engine~\cite{bouton2021atomic}. Finite-time performance must therefore be characterized not only by the mean work output, but also by its reliability and by the correlations that accumulate during repeated operation.

Population inversion and collective coupling can both enhance engine performance, but through different mechanisms. Population inversion changes the thermodynamic state supplied to the working medium. For a bounded spectrum it gives access to effective negative temperatures and can increase the gross output of an Otto cycle, as demonstrated in spin and cold-atom experiments and explored in finite-time and many-body settings~\cite{deassis2019efficiency,nettersheim2022power,deSousa2024,DamasInverted2026,brollo2025negative}. Collective coupling instead changes the dynamics: when many constituents exchange energy through a common environmental channel, their transition amplitudes combine coherently, modifying relaxation rates and producing power and reliability behavior that can differ from that of independent constituents~\cite{kloc2019collective,watanabe2020collective,souza2022collective,jaseem2023quadratic}. This mechanism originates from Dicke cooperative radiation~\cite{dicke1954coherence,gross1982superradiance} and underlies several superradiant and Dicke-inspired heat-engine proposals~\cite{hardal2015superradiant,niedenzukurizki2018collective,xu2024universal}.

A larger gross output from an inverted state does not, by itself, establish a thermodynamic advantage. Population-inverted states are active: part of their energy is already available as unitary-extractable work. Passivity and ergotropy quantify this state-level resource~\cite{pusz1978passive,allahverdyan2004maximal,allahverdyan2005minimal}, while nonequilibrium free-energy differences provide reversible formation bounds once the reference state and allowed operations are specified~\cite{brandao2013resource,horodecki2013fundamental}. For nonthermal reservoirs, the heat delivered to the working medium and the work required to prepare, maintain, or restore the reservoir are therefore distinct thermodynamic quantities~\cite{niedenzu2018bound,aguilar2022linear}. The central question is not simply whether inversion increases the extracted work, but how much of that increase reflects the active resource already stored in the inverted state and how much arises from collective finite-time dynamics.

We address this question with a symmetric collective spin formed by $N$ identical two-level
constituents. The compression and expansion strokes change the energy spacing while preserving the
populations of the Dicke states, whereas the hot and cold contacts drive collective transitions
between neighboring rungs of the ladder. Starting from a state diagonal in the Dicke basis, the
entire cycle remains diagonal and the quantum master equation reduces exactly to a finite
birth--death process on $N+1$ levels. The reservoirs determine the bias between upward and downward
transitions, while the Dicke matrix elements determine how rapidly population moves between
neighboring rungs. Passive and inverted operation can therefore be compared within the same
collective dynamics, with the work protocol unchanged.

This exact reduction places the engine in a standard statistical-mechanics setting. The stationary
cycle is the fixed point of a finite Markov process, finite-time relaxation is generated by
stochastic propagators, work and heat fluctuations follow from positive trajectory probabilities,
and repeated-cycle correlations can be obtained from tilted maps. Mean thermodynamics,
fluctuations, collective relaxation, and state-level resource accounting can therefore be treated
within the same finite-state description. For finite system size, positive contact durations, and
strictly positive transition rates, the population dynamics has a unique stationary cycle. In the
engine regime, the efficiency retains the standard Otto value fixed by the ratio of the two energy
gaps, while collective relaxation determines how much population is transferred during a finite
contact and therefore controls the work and power.

The complete-reset limit, in which each contact fully relaxes the working medium to its stationary
state, makes the thermodynamic effect of inversion especially transparent. For fixed reservoir
conditions away from infinite temperature, the cold and passive-hot populations remain localized
near the same end of the Dicke ladder as the system grows. Their mean positions remain separated by
only a finite number of rungs, and the passive work per cycle therefore saturates. An inverted hot
state is localized near the opposite end of the ladder. The separation between the cold and hot
populations then grows in proportion to the number of constituents, and so does the extracted
work. In asymptotic terms, the two branches display $O(1)$ and $O(N)$ complete-reset work,
respectively.

The same finite-ladder geometry controls the work fluctuations. In the complete-reset limit, the
Dicke partition function generates the work cumulants in closed form, following the finite-ladder
framework for homothetic Otto spectra developed in Ref.~\cite{DamasHomothetic2026}. For matched
passive and inverted hot states, the local widths of the endpoint distributions are the same even
though their mean positions can be separated across almost the entire ladder. The work reliability,
defined by the mean output relative to its fluctuations, therefore remains bounded on the passive
branch but grows linearly with the number of constituents on the inverted branch. By comparison,
combining $N$ statistically independent engines produces only the familiar square-root improvement
associated with statistical self-averaging. The enhanced reliability of the inverted collective
engine thus comes from increasing the signal without a corresponding growth of the endpoint
fluctuations.

This extensive complete-reset work is an endpoint effect, whereas collectivity controls the
finite-time kinetics. The Dicke matrix elements reshape the transition rates and can accelerate
relaxation during finite contacts. Near one end of the ladder, this dynamics admits a controlled
Holstein--Primakoff description. An inverted cycle is more demanding because each hot contact must
move population from the vicinity of one end of the ladder toward the other, over a distance that
itself grows with system size. Its finite-time behavior therefore cannot be inferred from a single
pole-local approximation and must be obtained from the full finite-state dynamics.

The resource analysis provides a complementary interpretation of the enhanced inverted output. For
matched passive and inverted hot states, the additional gross work is exactly the Otto efficiency
times the ergotropy carried by the inverted state. If the corresponding reversible excess
state-formation cost is charged once per cycle, the incremental balance relative to the matched
passive engine becomes negative. The enhanced gross output therefore represents the conversion of a
pre-existing active-state resource rather than a free energetic gain. The same conclusion holds for the passive comparators considered here, provided each operates as an engine and uses the same reversible state-level accounting with a complete reset to the cold reference.

Finite contacts add another layer of statistical behavior. An incomplete contact retains
information about the population entering the stroke, producing correlations between the states at
different corners of the same cycle. The final state of one cycle can also influence the next, so
a stationary mean output does not imply statistically independent operation. We evaluate these
finite-contact statistics directly from the finite-dimensional propagators and positive trajectory
probabilities, and use tilted cycle maps to distinguish one-cycle fluctuations from the long-time
fluctuations accumulated over repeated operation. Holstein--Primakoff and Kramers--Moyal
descriptions are used only as controlled large-system diagnostics in their respective domains,
while symmetry-leakage and secular-control criteria quantify the range over which the ideal
collective description remains reliable.

The remainder of the paper is organized as follows. Section~\ref{sec:model} introduces the
collective working medium, reservoir dynamics, and four-stroke protocol.
Section~\ref{sec:means} derives the stationary population cycle and its mean thermodynamics.
Section~\ref{sec:fluctuations} develops the complete-reset work statistics and reliability, and
Sec.~\ref{sec:cost} analyzes the thermodynamic resource carried by the inverted state.
Finite-contact trajectory statistics, repeated-cycle correlations, and the approximation hierarchy
are presented in Sec.~\ref{sec:finite-fluct}, followed by the conclusions and outlook in
Sec.~\ref{sec:conclusion}.

\section{Model and Four-Stroke Protocol}
\label{sec:model}

We now specify the symmetric collective spin, commuting gap modulation, and collective dissipative contacts used throughout the cycle. Unless stated otherwise, all dynamical and thermodynamic quantities considered below (including Gibbs states, passivity, and ergotropy) are defined within the $(N+1)$-dimensional, permutation-symmetric $j=N/2$ working-medium Hilbert space. Benchmarks involving $N$ independent constituents and product-state references in the full $2^N$-dimensional Hilbert space are introduced in Secs.~IV.D and V.B, while symmetry-breaking leakage out of the symmetric sector is considered separately in Sec.~VI.F.

%============================================================
\subsection{Notation, energy levels, and temperature}
\label{sec:model-notation}

We use units with $\hbar=k_{\rm B}=1$. The symbol $N$ denotes the number of identical two-level
constituents, $j=N/2$ is the maximum collective-spin quantum number, and $m$ is the eigenvalue of
$J_z$ in the range $-j\leq m\leq j$. The control parameters $\Omega_c$ and $\Omega_h$ are positive
energy gaps, with $0<\Omega_c<\Omega_h$; throughout,
$\Delta\Omega\equiv\Omega_h-\Omega_c>0$.

Each reservoir is characterized by an inverse-temperature parameter $\beta_\alpha=1/T_\alpha$,
where $\alpha=c,h$. For the finite-dimensional working medium, $\beta_\alpha$ may be positive,
zero, or negative, and we use $x_\alpha\equiv\beta_\alpha\Omega_\alpha$. For a fixed positive
gap, $x_\alpha>0$ gives a passive Gibbs state, $x_\alpha=0$ gives the uniform infinite-temperature
state, and $x_\alpha<0$ gives a population-inverted Gibbs state. The sign of $x$ is therefore a
statement about the ordering of populations, not about a negative energy gap. For a bounded spectrum,
the entropy can decrease after reaching its maximum as the energy is increased further, so the
thermodynamic relation $\beta=\partial S/\partial U$ permits $\beta<0$. Such states are hotter
than any positive-temperature Gibbs state in the standard thermodynamic ordering, rather than colder
than absolute zero~\cite{purcellpound1951,ramsey1956thermodynamics}.

For the Hamiltonian used below, the energy of $|j,m\rangle$ is $E_\alpha(m)=-\Omega_\alpha m$.
Thus $m=j$ is the lowest-energy level and $m=-j$ is the highest-energy level. A passive state
places more probability near $m=j$, whereas an inverted state places more probability near
$m=-j$. Equivalently, the integer $k\equiv j-m=0,1,\ldots,N$ counts the distance from the
low-energy pole of the symmetric ladder. With the present Hamiltonian convention it can be viewed
as the number of spin excitations above that pole. Thus a passive state occupies only the first few
rungs ($k=O(1)$), whereas an inverted state accumulates population near the opposite end
($k\simeq N$). This energy ordering is useful when interpreting the signs of work and heat.

%============================================================
\subsection{Collective operators and the symmetric Dicke sector}
\label{sec:model-dicke}

The working medium is a collection of $N$ identical two-level constituents, each described by
Pauli operators $\sigma_\nu^{(r)}$ ($r=1,\ldots,N$). We define
\begin{equation}
    J_\nu\equiv\frac12\sum_{r=1}^N\sigma_\nu^{(r)}, \qquad
    J_\pm\equiv J_x\pm iJ_y=\sum_{r=1}^N\sigma_\pm^{(r)},
    \label{eq:Jnu}
\end{equation}
which satisfy $[J_z,J_\pm]=\pm J_\pm$ and $[J_+,J_-]=2J_z$. The maximum-spin sector has $j=N/2$ and $m=-j,-j+1,\ldots,j$. It is spanned by the $N+1$
permutation-symmetric Dicke states $|j,m\rangle$, and all dynamics
below are restricted to this subspace.

Within this sector, the ladder operators act as
\begin{equation}
    J_\pm|j,m\rangle=\sqrt{(j\mp m)(j\pm m+1)}\,|j,m\pm1\rangle.
    \label{eq:ladder}
\end{equation}

The squared matrix elements in Eq.~\eqref{eq:ladder} are the source of the collective dependence
of the transition rates. They are largest near the center of the ladder and vanish at its ends.
The physical distinction from independent relaxation is that a common reservoir couples to the
coherent sum $J_\pm=\sum_r\sigma_\pm^{(r)}$. The constituent amplitudes therefore combine
before the transition probability is formed, producing the Dicke factors
$(j\mp m)(j\pm m+1)$. With independent local reservoirs, the individual transition
probabilities would instead add without this collective ladder enhancement.

The sector is preserved when the Hamiltonian and all jump operators are built from $J_z$ and
$J_\pm$. Since $J^2$ commutes with each of these operators, the projector onto fixed $j$ commutes
with the evolution. We take the working medium to be prepared in the $j=N/2$ sector; the effect of
symmetry-breaking local channels is quantified in Sec.~\ref{sec:validity}.

%============================================================
\subsection{Hamiltonians and the four strokes}
\label{sec:model-otto}

The control Hamiltonian during an isochore at branch $\alpha$ is
\begin{equation}
    H_\alpha=-\Omega_\alpha J_z, \qquad \alpha\in\{c,h\}.
    \label{eq:Halpha}
\end{equation}

Since $H_h$ is proportional to $H_c$, the two work Hamiltonians commute. The work strokes can
therefore be implemented as a change of the gap without changing the Dicke populations. The cycle
corners are labeled as follows:

\begin{enumerate}
    \item corner 1: before compression, with population vector $\vec p_1$ and Hamiltonian $H_c$;
    \item corner 2: after compression, with the same populations and Hamiltonian $H_h$;
    \item corner 3: after the hot contact, with population vector $\vec p_3$ and Hamiltonian
    $H_h$;
    \item corner 4: after expansion, with the same populations and Hamiltonian $H_c$;
    \item corner 5: after the cold contact, with population vector $\vec p_1'$.
\end{enumerate}

The stationary cycle satisfies $\vec p_1'=\vec p_1$. The hot and cold contacts are the only
strokes that change the population vector. Figure~\ref{fig01} shows this sequence.

\begin{figure}[tbp]
    \centering
    \includegraphics[width=0.9\linewidth]{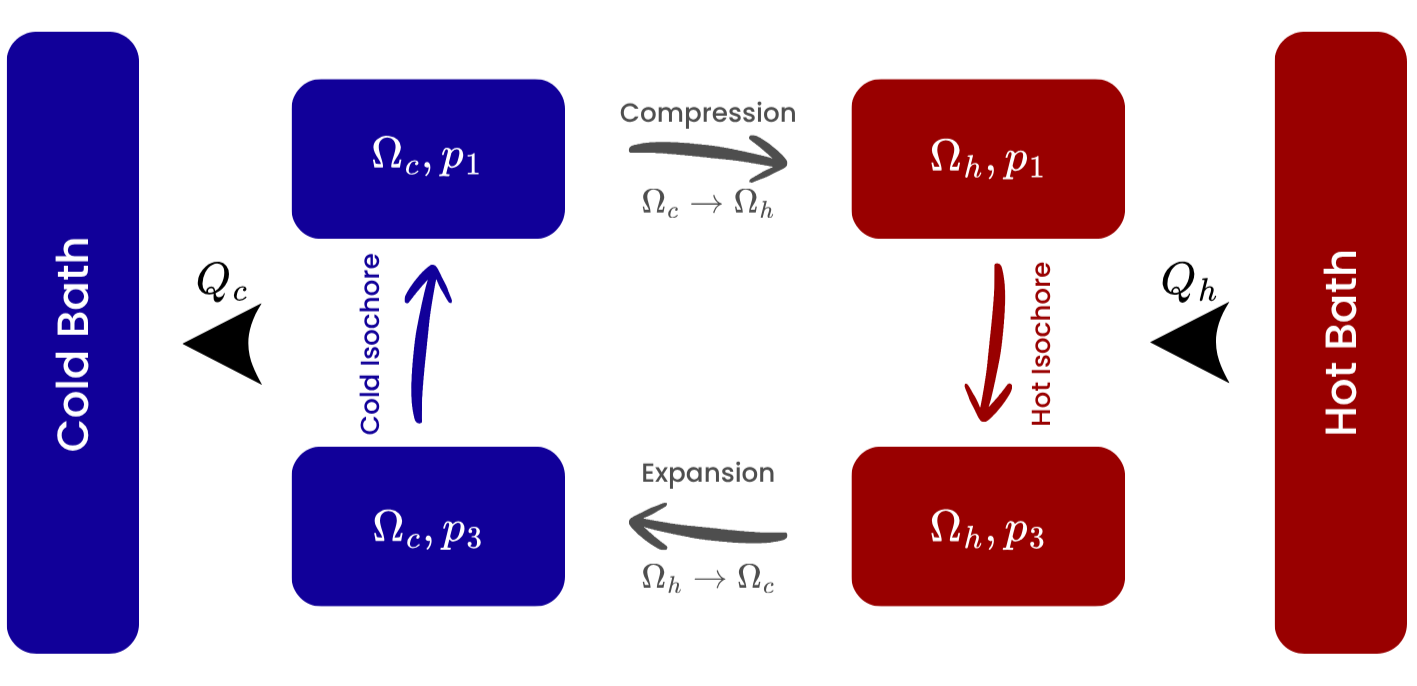}
    \caption{Four-stroke collective Otto cycle. The compression and expansion change the gap while
    preserving the Dicke populations. The hot and cold contacts act at fixed gaps and move the
    populations toward their respective stationary distributions. The stationary-cycle condition
    identifies corner 5 with corner 1.}
    \label{fig01}
\end{figure}

Operationally, the protocol requires three experimental capabilities: tuning the level spacing
between $\Omega_c$ and $\Omega_h$, switching the engineered hot and cold contacts, and performing
cycle-resolved readout of the collective polarization at selected corners. Because the work strokes
preserve $m$, the recorded corner labels directly determine the corresponding energy changes;
repeated realizations then reconstruct the joint work and heat statistics considered below.

We define work and heat as energy changes of the working medium: positive work means energy is
provided to the working medium, and positive heat means energy enters it from a reservoir. Over one
stationary cycle,
\begin{equation}
    \langle W\rangle+\langle Q_h\rangle+\langle Q_c\rangle=0.
    \label{eq:first-law-cycle}
\end{equation}

An engine has $\langle W\rangle<0$, $\langle Q_h\rangle>0$, and $\langle Q_c\rangle<0$. We report
the extracted gross power and efficiency as
\begin{equation}
    P_{\gross}\equiv\frac{-\langle W\rangle}{\tau_{\rm cyc}},\qquad
    \eta\equiv\frac{-\langle W\rangle}{\langle Q_h\rangle},
    \label{eq:P-eta-gross}
\end{equation}
when the cycle operates in the engine regime. Here
$\tau_{\rm cyc}=\tau_h+\tau_c+\tau_{u,1}+\tau_{u,2}$ is the duration of one cycle. The numerical
results set the two unitary-stroke durations to zero; adding finite durations changes the power
denominator but not the endpoint energy changes in this commuting model.

%============================================================
\subsection{Collective reservoir dynamics}
\label{sec:model-reservoir}

Collective coupling means that all constituents exchange energy with the same bath operator,
\begin{equation}
    H_{SB}=(J_++J_-)\otimes B_\alpha,
    \label{eq:HSB}
\end{equation}
rather than through independent local operators
$\sum_r(\sigma_+^{(r)}+\sigma_-^{(r)})\otimes B_\alpha^{(r)}$. The independent-copy benchmark is built from the corresponding local constituent dynamics.

"Davies-form" refers to the frequency-resolved Lindblad structure of the reduced generator. The collective Davies-form generator is
\cite{davies1974markovian,breuer2002theory}
\begin{equation}
    \dot\rho=\mathcal L_\alpha\rho=-i[H_\alpha,\rho]
    +\Gamma_{\downarrow,\alpha}\mathcal D[J_+]\rho
    +\Gamma_{\uparrow,\alpha}\mathcal D[J_-]\rho,
    \label{eq:Lalpha}
\end{equation}
where $\mathcal D[L]\rho\equiv L\rho L^\dagger-\tfrac12\{L^\dagger L,\rho\}$.
$J_+$ lowers the energy because it increases $m$, while $J_-$ raises the energy. $\Gamma_{\downarrow,\alpha}$ is the elementary downward rate and $\Gamma_{\uparrow,\alpha}$ is the elementary upward rate.

For a positive-temperature thermal Davies generator, the Kubo--Martin--Schwinger (KMS) detailed-balance condition gives the rate ratio
~\cite{davies1974markovian,esposito2009nonequilibrium}
\begin{equation}
    \frac{\Gamma_{\uparrow,\alpha}}{\Gamma_{\downarrow,\alpha}}
    =e^{-\beta_\alpha\Omega_\alpha}=e^{-x_\alpha}.
    \label{eq:kms}
\end{equation}

For the engineered inverted branch, the same detailed-balance ratio parametrizes an active stationary channel: $x_\alpha<0$ makes the upward elementary rate larger than the downward one. The inverted contact is an engineered active channel, not an ordinary positive-temperature equilibrium reservoir. The preparation resource associated with that active channel is treated in Sec.~\ref{sec:cost}.

To compare passive and inverted contacts at the same elementary kinetic scale, the matched-total-rate convention is imposed:
\begin{equation}
\Gamma_{\downarrow,\alpha}+\Gamma_{\uparrow,\alpha}=\gamma, \qquad
\begin{aligned}
\Gamma_{\downarrow,\alpha}&=\frac{\gamma}{1+e^{-x_\alpha}}, &
\\\Gamma_{\uparrow,\alpha}&=\frac{\gamma e^{-x_\alpha}}{1+e^{-x_\alpha}}.
\end{aligned}
\label{eq:matched-rate}
\end{equation}
\begin{table}[b]
\centering
\caption{Rate and state-matching conventions used in the finite-contact comparisons.}
\label{tab:rate-conventions}
\begin{tabular}{@{}p{0.25\linewidth}p{0.66\linewidth}@{}}
\toprule
Comparison & What is held fixed and what it tests \\
\midrule
Endpoint matching & The signed parameters $x_c$ and $x_h$ are fixed. This isolates the effect of
the endpoint populations and is the convention for the complete-reset formulas. \\
Matched total rate & The sum $\Gamma_{\downarrow}+\Gamma_{\uparrow}=\gamma$ is fixed. This
compares passive and inverted contacts with the same elementary kinetic scale. \\
Dynamical matching & A measured relaxation time or Liouvillian gap is fixed. This is the
appropriate control when two physical reservoirs have different microscopic rate amplitudes. \\
\bottomrule
\end{tabular}
\end{table}

The parameter $\gamma$ fixes the elementary relaxation scale, while the collective $N$ dependence
enters through the matrix elements in Eq.~\eqref{eq:ladder}.

The weak-coupling Davies-form reduction assumes weak system--bath coupling, a bath correlation time
short compared with the relaxation time, and sufficient separation of relevant Bohr frequencies
\cite{davies1974markovian}. The ratio in Eq.~\eqref{eq:kms} determines the fixed point, whereas a
finite-contact prediction also depends on the rate convention summarized in
Table~\ref{tab:rate-conventions}.

All finite-contact figures reported below use the matched-total-rate convention unless a different choice is stated explicitly. For the secular-controlled finite-contact benchmark we additionally scale the elementary rate and contact duration at fixed exposure $u=\gamma\tau$ so that the collective rates
satisfy the explicit secular-control criterion introduced in Sec.~\ref{sec:validity}. The complete-reset results do not depend on this convention because the contacts have erased their incoming-state memory.

The model therefore combines three distinct choices: restriction to one permutation-symmetric
sector, a Markovian Davies-form reservoir description, and the matched-total-rate convention for
finite-contact comparisons. Numerical path truncation and the large-$N$ approximations introduced
later are separate approximations.

\subsection{Anti-passivity: the broader finite-dimensional mechanism}
\label{sec:anti-passivity}

Although the engine studied below has commuting work strokes, it is useful first to isolate the
more general finite-dimensional reason why an inverted population can increase extracted work,
independently of the special language of negative temperature.
Consider a finite Hamiltonian with ordered energies
$\epsilon_0\leq\epsilon_1\leq\cdots\leq\epsilon_{d-1}$ and a population vector
$\boldsymbol p$ written in the same order. A population is called \emph{passive} when it
decreases as the energy increases, and \emph{anti-passive} when the ordering is reversed. A
unitary stroke between two energy eigenbases induces the transition matrix
$\mathsf P_{m|n}=|\langle m_f|U|n_i\rangle|^2$. This matrix is doubly stochastic: its rows and
columns each sum to one. The matrix need not be a permutation, because a finite-time stroke can
spread one initial energy population over several final levels.

Define the energy defect of this stroke by
\begin{equation}
    \delta E(\mathsf P,\boldsymbol p)
    \equiv\boldsymbol\epsilon^{\mathsf T}
    (\mathsf P\boldsymbol p-\boldsymbol p).
    \label{eq:anti-passivity-defect}
\end{equation}

It is the energy increase produced by the stroke when the same ordered spectrum is used as the
reference. The following sign rule then follows.

\begin{proposition}[Passivity sign rule]
\label{prop:anti-passivity}
For a passive population, $\delta E(\mathsf P,\boldsymbol p)\geq0$ for every doubly
stochastic $\mathsf P$. For an anti-passive population, the inequality is reversed.
\end{proposition}

\begin{proof}
By the Birkhoff--von Neumann theorem, every doubly stochastic matrix is a convex combination of permutation
matrices. The rearrangement inequality states that the passive ordering minimizes
$\boldsymbol\epsilon^{\mathsf T}\Pi\boldsymbol p$ over permutations $\Pi$, whereas the
anti-passive ordering maximizes it. Taking the same convex combination preserves the corresponding
inequality. This is the finite-dimensional form of the minimum-work/passivity principle
\cite{pusz1978passive,allahverdyan2004maximal,allahverdyan2005minimal}.
\end{proof}

Physically, the statement is simple. Mixing or redistributing a passive population cannot lower
its energy below the passive ordering, whereas the same redistribution can lower the energy of an
anti-passive population. The latter therefore contains unitary-extractable work.

For a homothetic Otto pair, in which the low-gap spectrum is a uniform rescaling of the high-gap
spectrum, let the high-gap energies be $\boldsymbol\epsilon$ and the low-gap energies be
$\alpha\boldsymbol\epsilon$, with $0<\alpha<1$. If $\mathsf P_e$ and
$\mathsf P_c$ are the expansion and compression transition matrices, respectively, the extracted
work is
\begin{equation}
\begin{split}
 W_{\rm out}={}&\boldsymbol\epsilon^{\mathsf T}\boldsymbol p_h
 -\alpha\boldsymbol\epsilon^{\mathsf T}\mathsf P_e\boldsymbol p_h
 +\alpha\boldsymbol\epsilon^{\mathsf T}\boldsymbol p_c
 -\boldsymbol\epsilon^{\mathsf T}\mathsf P_c\boldsymbol p_c .
\end{split}
\label{eq:homothetic-work}
\end{equation}

The adiabatic reference has $\mathsf P_e=\mathsf P_c=\mathbb I$ and
$W_{\rm out}^{\rm ad}=(1-\alpha)\boldsymbol\epsilon^{\mathsf T}
(\boldsymbol p_h-\boldsymbol p_c)$. Defining
$\delta E_e=\delta E(\mathsf P_e,\boldsymbol p_h)$ and
$\delta E_c=\delta E(\mathsf P_c,\boldsymbol p_c)$ gives the exact comparison
\begin{equation}
    W_{\rm out}-W_{\rm out}^{\rm ad}
    =-\alpha\delta E_e-\delta E_c.
    \label{eq:homothetic-defect}
\end{equation}

Thus an anti-passive hot endpoint and a passive cold endpoint improve the adiabatic output only
when the energy decrease during expansion exceeds the compression defect after multiplication by
$\alpha$. For the commuting work strokes used here, $\mathsf P_e=\mathsf P_c=\mathbb I$ and
both defects vanish. The inverted-branch enhancement therefore originates entirely from the endpoint
populations and their resource content; nonadiabatic-friction effects are absent from the cycle
considered below.

The next question is quantitative: how much of that active-state resource can appear as
additional Otto work when the same finite-time work stroke is applied to the passive and active
states? The same comparison gives a useful bound on the differential value of an active state. Let
$\boldsymbol q$ be an anti-passive population and let $\boldsymbol\pi$ be its passive rearrangement,
so that both vectors have the same eigenvalues and
\begin{equation}
    \mathcal W_{\rm erg}
    \equiv \boldsymbol\epsilon^{\mathsf T}
    (\boldsymbol q-\boldsymbol\pi)>0
    \label{eq:general-ergotropy}
\end{equation}
is the ergotropy on the reference spectrum. Compare two cycles that use the same expansion map
$\mathsf P$ and differ only in the hot population, $\boldsymbol q$ versus
$\boldsymbol\pi$. Their extracted-work difference is
\begin{equation}
    \Delta W_{\rm act-pass}
    =\mathcal W_{\rm erg}\bigl(1-\alpha r_{\mathsf P}\bigr),
    \qquad
    r_{\mathsf P}
    \equiv\frac{\boldsymbol\epsilon^{\mathsf T}\mathsf P
    (\boldsymbol q-\boldsymbol\pi)}{\mathcal W_{\rm erg}}.
    \label{eq:differential-gain}
\end{equation}

\begin{proposition}[Tight active--passive differential bound]
\label{prop:differential-bound}
For every doubly stochastic expansion map $\mathsf P$,
$-1\leq r_{\mathsf P}\leq1$. Consequently, the differential gain factor
$\chi\equiv\Delta W_{\rm act-pass}/\mathcal W_{\rm erg}$ satisfies
\begin{equation}
    1-\alpha\leq\chi\leq1+\alpha.
    \label{eq:differential-bound}
\end{equation}

The lower and upper equalities are attained by the identity and reversal maps, respectively.
\end{proposition}

\begin{proof}
Because $\boldsymbol\pi$ is a permutation of $\boldsymbol q$, the energy of either population
after a doubly stochastic map lies between the passive and anti-passive rearrangement energies.
Their difference is therefore bounded in magnitude by
$\boldsymbol\epsilon^{\mathsf T}(\boldsymbol q-\boldsymbol\pi)=\mathcal W_{\rm erg}$,
which proves $|r_{\mathsf P}|\leq1$. Substitution into Eq.~\eqref{eq:differential-gain}
gives Eq.~\eqref{eq:differential-bound}.
\end{proof}

The factor $\chi$ compares the work value of two hot states at fixed work-stroke map; preparation
and renewal costs enter separately. For the commuting cycle $\mathsf P=\mathbb I$, so
$r_{\mathsf P}=1$ and
$\Delta W_{\rm act-pass}=(1-\alpha)\mathcal W_{\rm erg}$. This recovers the matched
complete-reset result below after identifying $\alpha=\Omega_c/\Omega_h$.

Optimizing an active state at fixed relative-entropy athermality instead asks which population
distribution maximizes a chosen energy objective under a fixed formation budget
~\cite{brandao2013resource,horodecki2013fundamental}. Restricting the analysis to the Gibbs-Dicke
family isolates the collective matrix elements, inversion symmetry, and resource comparison of the
engine studied here. Equations~\eqref{eq:differential-gain}--\eqref{eq:differential-bound} then
set the ergotropy-scaled interval for the active--passive work difference under a general doubly
stochastic work-stroke map.

\section{Exact Statistics of the Stationary Cycle}
\label{sec:means}

%============================================================
\subsection{Reduction to Populations and the Master Equation}
\label{sec:means-reduction}

For Dicke-diagonal initial states, the collective Davies-form generator closes on the populations
$p(m,t)\equiv\langle j,m|\rho(t)|j,m\rangle$. The vector
$\vec p=(p(-j),p(-j+1),\ldots,p(j))^{\mathsf T}$ has $N+1$ nonnegative real entries and is
normalized to one, reducing the $(N+1)^2$ density-matrix variables to a classical population
process.

Starting from a state diagonal in the basis $\{|j,m\rangle\}$, the generator $\mathcal L_\alpha$
itself preserves diagonality. This follows from a direct calculation: since $J_+|j,m\rangle$ and
$J_-|j,m\rangle$ are each proportional to a single basis state, the term $J_\pm\rho J_\pm^\dagger$
of a diagonal $\rho$ remains diagonal; and since $J_+^\dagger J_+=J_-J_+$ and $J_-^\dagger
J_-=J_+J_-$ are functions exclusively of $J_z$, and therefore also diagonal in this basis, the
anticommutator of the Lindblad dissipator preserves diagonality in the same way.

The unitary contribution also vanishes on a Dicke-diagonal state because $H_\alpha$ is diagonal in the same basis. Since the collective ladder operators $J_\pm$ connect only neighboring Dicke states, the population dynamics is a one-dimensional continuous-time birth--death process on the rungs $m=-j,\ldots,j$: from a given rung $m$, population can jump only to $m+1$ or $m-1$. Equivalently, the population generator is tridiagonal in the Dicke basis. The master equation is
\begin{equation}
\begin{aligned}
\dot p(m)={}&\Gamma_{\downarrow,\alpha}(j-m+1)(j+m)p(m-1)\\
&+\Gamma_{\uparrow,\alpha}(j+m+1)(j-m)p(m+1)\\
&-\Gamma_{\downarrow,\alpha}(j-m)(j+m+1)p(m)\\
&-\Gamma_{\uparrow,\alpha}(j+m)(j-m+1)p(m).
\end{aligned}
\label{eq:master-pop}
\end{equation}

The first two terms are probability currents into rung $m$ from its two neighbors, whereas the last two terms describe probability leaving $m$ toward those neighbors. The Dicke matrix elements set the rung-dependent transition factors, while
$\Gamma_{\downarrow,\alpha}$ and $\Gamma_{\uparrow,\alpha}$ set the elementary
bath-induced downward and upward energy-transfer rates.

This birth--death equation can be written in the usual vector form
$\dot{\vec p}=R_\alpha\vec p$ as
\begin{equation}
\begin{split}
    [R_\alpha]_{m+1,m} = \Gamma_{\downarrow,\alpha}\,(j-m)(j+m+1), \\
    [R_\alpha]_{m-1,m} = \Gamma_{\uparrow,\alpha}\,(j+m)(j-m+1),
\end{split}
\label{eq:Ralpha}
\end{equation}
with the diagonal entries fixed by conservation of probability.

The two off-diagonal entries in Eq.~\eqref{eq:Ralpha} have a direct interpretation. The rate from
$m$ to $m+1$ is $\Gamma_{\downarrow,\alpha}(j-m)(j+m+1)$, and the rate from $m$ to $m-1$ is
$\Gamma_{\uparrow,\alpha}(j+m)(j-m+1)$. The factors multiplying the elementary bath rates are
the squared collective matrix elements. Thus the bath fixes the upward/downward bias, while the
Dicke ladder fixes how that bias is amplified at each rung.

%============================================================
\subsection{Limit Cycle, Average Work and Heat}
\label{sec:means-cycle}

Let $\vec p_1$ denote the stationary population immediately before compression and $\vec p_3$ the
population immediately after the hot contact. Each isochore is represented by the nonnegative,
column-stochastic propagator $K_\alpha\equiv e^{R_\alpha\tau_\alpha}$. Because the unitary
strokes preserve the populations, $\vec p_3=K_h\vec p_1$ and
$\vec p_1^{\,\prime}=K_c\vec p_3$, where $\vec p_1^{\,\prime}$ denotes the population at the start of the following cycle. The
limit-cycle condition, $\vec p_1^{\,\prime}=\vec p_1$, closes the problem into the fixed-point
equation $\vec p_1=K_cK_h\vec p_1$.

For strictly positive upward and downward rates and positive contact durations, every neighboring
transition is dynamically accessible and each contact propagator is strictly positive. The product
$K_cK_h$ is therefore a primitive stochastic matrix. By the Perron--Frobenius theorem, it has a
simple dominant eigenvalue at $1$ and a unique, strictly positive eigenvector, to which any initial
population converges geometrically. The limit cycle $\vec p_1$ therefore exists and is unique.

Any observable diagonal in $J_z$ reduces to an average over this same distribution $\vec p$. In
particular, we define the average magnetization as
$\langle m\rangle_{\vec p}\equiv\sum_{m=-j}^{j}m\,p(m)=\langle J_z\rangle$, so that the internal
energy at any point of the cycle is
$U=\mathrm{Tr}(H_\alpha\rho)=-\Omega_\alpha\langle J_z\rangle=-\Omega_\alpha\langle m\rangle_{\vec p}$.

The work and heat exchanged at each stroke follow directly from how $U$ varies, that is, through
the change in $\Omega_\alpha$ during the unitary strokes and through the relaxation of the
population during the isochores.

For example, the compression changes the energy of a state with label $m$ by
$[-\Omega_hm]-[-\Omega_cm]=-\Delta\Omega m$. Averaging with the corner-1 distribution gives
$\langle W_{1\to2}\rangle=-\Delta\Omega M_1$. The expansion gives
$\langle W_{3\to4}\rangle=+\Delta\Omega M_3$. The hot contact changes the energy at fixed
$\Omega_h$, and the cold contact changes it at fixed $\Omega_c$.

We therefore denote by $M_1\equiv\langle m\rangle_{\vec p_1}$ and $M_3\equiv\langle m\rangle_{\vec
p_3}$ the values of this magnetization at the two vertices of the cycle, and, going through the
four strokes in this order, we obtain at the stationary cycle
\begin{equation}
\begin{aligned}
\langle W\rangle&=\Delta\Omega\,(M_3-M_1),\\
\langle Q_h\rangle&=-\Omega_h(M_3-M_1),\\
\langle Q_c\rangle&=\Omega_c(M_3-M_1).
\end{aligned}
\label{eq:means-general}
\end{equation}

Thus all mean thermodynamic observables of the commuting cycle are controlled by a single
experimentally accessible quantity, the change $M_3-M_1$ in collective polarization across the
hot contact. Measuring the corner magnetizations therefore determines the mean work and both mean
heats without reconstructing the full density matrix. With
$\Delta\Omega\equiv\Omega_h-\Omega_c$, the cold-contact expression uses the stationary
condition: the mean magnetization after the cold contact is again $M_1$. Before imposing that
condition, the cold heat is $\Omega_c(M_3-M_5)$, where $M_5$ is the mean at the end of the cold
contact. Thus the average first law follows from $M_5=M_1$, whereas the corresponding pathwise
identity is given explicitly in Sec.~\ref{sec:finite-path}.

Substituting Eq.~\eqref{eq:means-general} into the efficiency definition cancels the common
magnetization difference,
\begin{equation}
    \eta = \frac{\Delta\Omega}{\Omega_h} = 1-\frac{\Omega_c}{\Omega_h},
    \label{eq:efficiency}
\end{equation}
for $0<\Omega_c<\Omega_h$ at engine operating points with $M_3\neq M_1$ and
$\langle Q_h\rangle>0$. The collective reservoir controls $M_3-M_1$ and therefore the work and
power,
\begin{equation}
    P_{\gross} = \frac{\Delta\Omega\,|M_3-M_1|}{\tau_{\rm cyc}}.
    \label{eq:power-general}
\end{equation}

For the zero-duration commuting work strokes used in the numerical calculations,
$\tau_{\rm cyc}=\tau_h+\tau_c$. The endpoint efficiency remains independent of $N$ and of the
contact durations in this commuting protocol.

%============================================================
\subsection{Gibbs-Dicke Thermodynamics}
\label{sec:means-gibbs}

In the complete-thermalization limit ($\tau_\alpha\to\infty$), $\vec p_1$ and $\vec p_3$ each relax
to the stationary state of the respective generator $R_\alpha$. For a birth--death chain, this
stationary state satisfies the detailed-balance condition
$[R_\alpha]_{m+1,m}p(m)=[R_\alpha]_{m,m+1}p(m+1)$. Substituting Eq.~\eqref{eq:Ralpha}, the factor
$(j-m)(j+m+1)$ cancels on both sides, leaving
$p(m+1)/p(m)=\Gamma_{\downarrow,\alpha}/\Gamma_{\uparrow,\alpha}=e^{\beta_\alpha\Omega_\alpha}$ by
the rate ratio in Eq.~\eqref{eq:kms}. Since this ratio does not depend on $m$, the stationary
state is the \emph{Gibbs-Dicke state},
\begin{equation}
    \rho_x \equiv \frac{e^{xJ_z}}{Z_j(x)}, \qquad p_x(m)=\frac{e^{xm}}{Z_j(x)}, \qquad
    x\equiv\beta_\alpha\Omega_\alpha.
    \label{eq:gibbs-dicke}
\end{equation}

At $x=0$ this expression is understood by continuity: $p_0(m)=1/(N+1)$ and
$Z_j(0)=N+1$.

The normalization is a finite geometric sum of $N+1$ terms, which can be evaluated in closed form as
\begin{equation}
    Z_j(x) = \sum_{m=-j}^j e^{xm} = \frac{\sinh[(N+1)x/2]}{\sinh(x/2)}.
    \label{eq:Zj}
\end{equation}

Defining $\mu_N(x)\equiv\partial_x\ln Z_j(x)$, the average magnetization in the Gibbs-Dicke state
is
\begin{equation}
    \mu_N(x) = \frac{N+1}{2}\coth\!\left(\frac{(N+1)x}{2}\right) - \frac12\coth\!\left(\frac x2
    \right).
    \label{eq:muN}
\end{equation}

In the limit $\tau_\alpha\to\infty$, therefore, $M_1\to\mu_N(x_c)$ and $M_3\to\mu_N(x_h)$, and
Eq.~\eqref{eq:means-general} becomes a closed-form expression in $N$, $x_h$, and $x_c$,
\begin{equation}
\begin{split}
    \langle W\rangle & = \Delta\Omega\big[\mu_N(x_h)-\mu_N(x_c)\big], \\
    \langle Q_h\rangle & = -\Omega_h\big[\mu_N(x_h)-\mu_N(x_c)\big], \\
    \langle Q_c\rangle & = \Omega_c\big[\mu_N(x_h)-\mu_N(x_c)\big].
\end{split}
\label{eq:means-complete-reset}
\end{equation}

These expressions define the complete-reset endpoint benchmark for the energy exchanged per
cycle. In the engine regime the extracted work per cycle is
\begin{equation}
    W_{\rm out}^{\rm reset}
    \equiv-\langle W\rangle
    =\Delta\Omega\big[\mu_N(x_c)-\mu_N(x_h)\big].
    \label{eq:work-complete-reset}
\end{equation}

In the literal dynamical limit $\tau_h,\tau_c\to\infty$, the power defined in
Eq.~\eqref{eq:P-eta-gross} vanishes because the cycle duration diverges. We therefore use the
complete-reset solution to state work-per-cycle and fluctuation results; power is evaluated at
finite contact times in Sec.~\ref{sec:means-nearpole} and Sec.~\ref{sec:finite-fluct}.

The efficiency, Eq.~\eqref{eq:efficiency}, remains unchanged in the complete-reset endpoint
benchmark, and it is the magnetization $\mu_N(x)$ that carries all of the dependence on $N$.

%============================================================
\subsection{Collective Scaling: Saturation and Extensivity}
\label{sec:means-scaling}

\begin{figure}[tbp]
    \centering
    \includegraphics[width=0.85\linewidth]{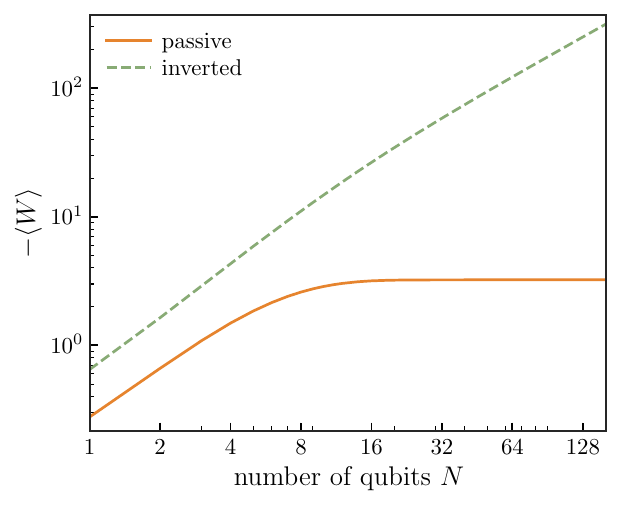}
    \caption{Complete-reset extracted work per cycle as a function of $N$ for passive and inverted
    hot branches, illustrating saturation, Eq.~\eqref{eq:W-saturation}, versus extensivity,
    Eq.~\eqref{eq:W-extensivity}. Parameters: $\Omega_c=1$, $\Omega_h=3$, $x_c=1$,
    $x_h=+0.375$ for the passive branch, and $x_h=-0.375$ for the inverted branch. Finite-contact
    power is shown separately in Fig.~\ref{fig:power-tau}.}
    \label{fig:saturation}
\end{figure}

It is natural to ask how $\mu_N(x)$ behaves as $N\to\infty$ with the dimensionless parameter $x$
fixed. This means that the dimensionless thermal bias $\beta\Omega$ is held fixed; it is not the
same as holding $\beta$ fixed while changing $\Omega$. Since
$\coth(y)\to\mathrm{sgn}(y)$ exponentially fast as $|y|\to\infty$, the first term of
Eq.~\eqref{eq:muN} saturates at $\pm N/2$ according to the sign of $x$, while the second term
remains finite; expanding both, one obtains
\begin{equation}
\begin{split}
    \mu_N(x) & = \frac N2 - \bar n(x) + O\big(Ne^{-N x}\big) \quad (x>0), \\
    \mu_N(x) & = -\frac N2 + \bar n(-x) + O\big(Ne^{-N|x|}\big) \quad (x<0),
\end{split}
\label{eq:muN-asymptotic}
\end{equation}
with $\bar n(x)\equiv1/(e^x-1)$ the same Bose--Einstein-type function from the thermodynamics of a
single bosonic mode.

Substituting into Eq.~\eqref{eq:work-complete-reset}, the term linear in $N$ cancels when the
two endpoint states are on the passive side of the spectrum, but it survives when the hot endpoint
is inverted. For an engine-compatible passive ordering $0<x_h<x_c$,
\begin{align}
W_{\rm out}^{\rm reset}
&\longrightarrow\Delta\Omega\left[\bar n(x_h)-\bar n(x_c)\right]\notag\\
&=O(1),\qquad 0<x_h<x_c,
\label{eq:W-saturation}\\
W_{\rm out}^{\rm reset}
&\longrightarrow\Delta\Omega\left[N-\bar n(-x_h)-\bar n(x_c)\right]\notag\\
&=O(N),\qquad x_h<0.
\label{eq:W-extensivity}
\end{align}

The extracted work per cycle therefore \emph{saturates} on the passive branch and grows
\emph{extensively} on the inverted branch, as illustrated in Fig.~\ref{fig:saturation}. 
These two scaling behaviors can be understood from the location of the endpoint populations 
on the Dicke ladder. For fixed $x > 0$, the Gibbs-Dicke distribution remains confined to an 
$\mathcal{O}(1)$-wide boundary layer near the low-energy pole $m = j$: writing $k = j - m$, 
one has $p(k) \propto e^{-xk}$ and $\langle k \rangle \to \bar{n}(x)$. Both cold and passive-hot 
endpoints remain near the same pole as $N$ grows, and increasing $N$ mainly adds distant rungs 
that carry negligible probability. Their mean magnetizations differ by only $\mathcal{O}(1)$, 
which explains the saturation of the passive work. 

For an inverted hot state, the hot distribution is instead localized within $\mathcal{O}(1)$ 
rungs of the opposite pole $m = -j$. The cold and hot distributions are then separated by 
essentially the full Dicke ladder, $\Delta m = N - \mathcal{O}(1)$, so the gap modulation 
converts an $\mathcal{O}(N)$ polarization displacement into extensive work. The inverted 
scaling therefore originates from the macroscopic separation of the endpoint populations, 
not from a broadening of either distribution. Endpoint extensivity is distinct from collective kinetic enhancement: the former follows from the location of the stationary populations, whereas the latter arises from the Dicke transition matrix elements and governs the power at finite contact time.

%============================================================
\subsection{Power at Fixed \texorpdfstring{$\tau$}{tau} Near a Pole}
\label{sec:means-nearpole}

At finite contact times, Eq.~\eqref{eq:power-general} depends on the stationary corner
magnetizations and generally has no closed form. There is, however, an analytically controlled
large-$N$ regime with $\tau_h,\tau_c$ fixed.

Let $k\equiv j-m$ be the distance to the north pole $m=j$. Rewriting Eq.~\eqref{eq:Ralpha} in
terms of $k$, the transition rates $k\to k-1$ (toward the pole) and $k\to k+1$ (away from it)
become $\Gamma_{\downarrow,\alpha}\,k(N-k+1)$ and $\Gamma_{\uparrow,\alpha}\,(N-k)(k+1)$,
respectively.

For a positive-$x_\alpha$ contact, for which the north pole is the stable endpoint, and for
$k=O(1)$ fixed as $N\to\infty$, both rates become \emph{linear} in $k$ at leading order, in
contrast with the original rates in $m$, which are quadratic and therefore unable to close the
equation for the mean. This local linearization allows one to show that $\langle k\rangle$ obeys
\begin{equation}
    \frac{d\langle k\rangle}{dt} = N\Gamma_{\uparrow,\alpha} -
    N\gamma\tanh\!\Big(\frac{x_\alpha}2\Big)\langle k\rangle,
    \label{eq:k-ode}
\end{equation}
whose solution, for an initial condition $\langle k(0)\rangle$, is
\begin{equation}
    \langle k(t)\rangle = \bar n(x_\alpha) + \big[\langle k(0)\rangle-\bar n(x_\alpha)\big]\,
    e^{-N\gamma\tanh(x_\alpha/2)\,t}.
    \label{eq:k-sol}
\end{equation}

Near a pole, therefore, the Dicke ladder reduces at leading order to a linear birth--death process,
equivalently the leading Holstein--Primakoff bosonic description, with a relaxation rate enhanced by
a factor of $N$ through the collective coupling to the common bath operator. In this regime the
collective enhancement is kinetic rather than energetic: the equilibrium distance from the pole
remains $O(1)$, but the time needed to relax that local displacement is shortened by the factor
$N$. This is why increasing $N$ can accelerate a passive contact even though the complete-reset
work of the passive branch itself does not become extensive.

For the passive branch ($x_c>0$ and $x_h>0$, both fixed), no stroke of the cycle moves the
population near the opposite pole; the two isochores relax populations concentrated near the same
pole $m=j$ throughout the entire protocol. By Eq.~\eqref{eq:k-sol}, the exponential factor
$e^{-N\gamma\tanh(x_\alpha/2)\tau_\alpha}\to0$ for any fixed $\tau_\alpha>0$, however small, as
long as $N$ is large enough; that is, each isochore reaches complete thermalization in this limit,
regardless of its duration. Substituting $\langle k\rangle_1\to\bar n(x_c)$ and $\langle
k\rangle_3\to\bar n(x_h)$ into Eq.~\eqref{eq:power-general}, one obtains, for $\tau_h=\tau_c=\tau$,
\begin{equation}
    P_{\text{passive}}^\infty(\tau) = \frac{\Delta\Omega}{2\tau}\big[\bar n(x_h)-\bar n(x_c)\big],
    \label{eq:power-infty-pass}
\end{equation}
valid for any fixed positive $\tau$ within the reduced birth--death model. A microscopic
Davies interpretation additionally requires the collective rates to remain small compared with the
Bohr frequencies. Figure~\ref{fig:power-tau} therefore uses the rate-scaled variables
$u=\gamma\tau$ and $P/\gamma$, which preserve the finite-contact maps under the controlled
rate--time rescaling specified below.

\begin{figure}[tbp]
    \centering
    \includegraphics[width=0.85\linewidth]{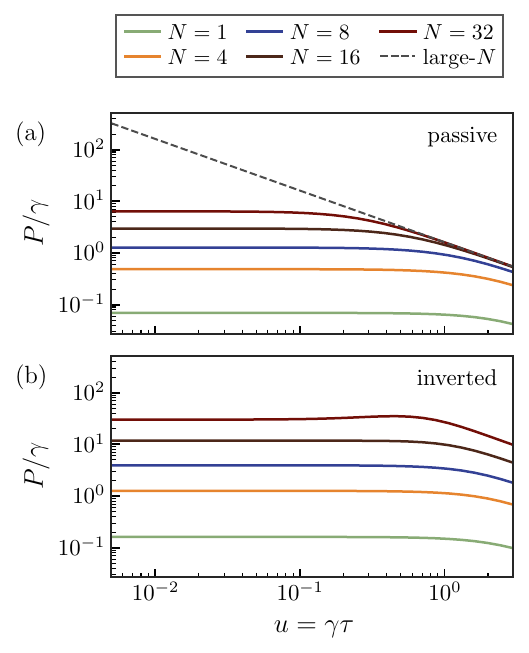}
    \caption{Rate-scaled gross power $P/\gamma$ as a function of the dimensionless contact exposure
    $u=\gamma\tau$ for passive and inverted hot branches and several values of $N$. The dashed
    reference curve in the passive panel is Eq.~\eqref{eq:power-infty-pass} divided by $\gamma$;
    the near-pole derivation does not cover the pole-to-pole traversal of the inverted branch.
    Parameters: $\Omega_c=1$, $\Omega_h=3$, $x_c=1$, $x_h=\pm0.375$, and
    $u\in[0.005,3]$. For the controlled benchmark $\gamma=4\times10^{-5}$, so
    $\tau=u/\gamma$.}
    \label{fig:power-tau}
\end{figure}

The same argument does not apply to the inverted branch. Since $x_h<0$, the hot isochore relaxes
the population toward the \emph{opposite} pole, $m=-j$, but the population entering this isochore
arrives concentrated near the pole $m=+j$, inherited from the cold isochore. At every cycle,
therefore, the population would need to cross the entire Dicke ladder, a distance of order $N$, and
Eq.~\eqref{eq:k-ode}, valid only for $k=O(1)$, does not describe this crossing. This is why no
closed-form formula analogous to Eq.~\eqref{eq:power-infty-pass} exists for the inverted branch at
fixed $\tau$ within the near-pole derivation. The complete-reset result,
Eq.~\eqref{eq:W-extensivity}, establishes extensive \emph{work per cycle} on the inverted branch,
but does not by itself determine the asymptotic finite-contact power. A coordinate transformation
about the opposite pole can give a local approximation there, but it does not solve the finite-time
traversal from one pole to the other. The near-pole formula is therefore restricted to the passive
fixed-contact limit established above.

\section{Complete-Reset Work and Hot-Heat Fluctuations}
\label{sec:fluctuations}

The mean work and heat depend only on the magnetizations at corners 1 and 3 after the stationary
cycle condition is imposed. Fluctuations require more information. A single cycle can begin at one
Dicke rung, leave the hot contact at another, and finish the cold contact at a third. The first two
labels are correlated whenever the hot contact is finite, while the third label is needed for the
actual cold-contact heat. We first derive the closed-form work and hot-heat statistics in the
complete-reset limit. Finite-contact joint work--heat statistics are treated separately in
Sec.~\ref{sec:finite-fluct}. The complete-reset analysis follows the standard two-point measurement (TPM) \cite{denzler2020efficiency,talkner2007fluctuation} and characteristic-function framework for quantum work statistics, specialized to the finite Dicke ladder, where commuting work strokes and Dicke-diagonal states make stochastic energy changes directly identifiable with the corner labels.

\subsection{Random variables and the joint distribution}
\label{sec:fluct-variables}

Let $m$ be the Dicke label at corner 1, $n$ the label at corner 3, and $r$ the label after the cold
contact at corner 5. Because the work strokes preserve the label, a trajectory with labels
$(m,n,r)$ has
\begin{equation}
\begin{aligned}
W&=\Delta\Omega(n-m),\\
Q_h&=-\Omega_h(n-m),\\
Q_c&=\Omega_c(n-r).
\end{aligned}
\label{eq:trajectory-WQ}
\end{equation}

The work and hot heat are tightly coupled on every trajectory, but the cold heat also depends on
the final label $r$. The pathwise first law is therefore
\begin{equation}
    W+Q_h+Q_c=\Omega_c(m-r)
    =E_c(r)-E_c(m).
    \label{eq:trajectory-first-law}
\end{equation}

It becomes the cycle first law after averaging in the stationary state, where
$\langle m\rangle=\langle r\rangle$. The joint probability of the three labels is
\begin{equation}
    \mathcal P(m,n,r)=p_1(m)[K_h]_{n,m}[K_c]_{r,n}.
    \label{eq:joint-finite-contact}
\end{equation}

The first factor describes how often the cycle starts at $m$, the second describes the hot-contact
transition, and the third describes the cold-contact transition. This joint distribution is
required whenever either contact is incomplete.

\subsection{Independence in the complete-reset limit}
\label{sec:fluct-independence}

For strictly positive upward and downward rates, the hot generator has one stationary distribution
and all other population modes decay. Consequently,
\begin{equation}
    [K_h]_{n,m}=[e^{R_h\tau_h}]_{n,m}
    \xrightarrow[\tau_h\to\infty]{}p_{x_h}(n),
    \label{eq:Kh-reset-limit}
\end{equation}
independently of the incoming label $m$. The cold contact similarly gives
$p_1(m)\to p_{x_c}(m)$. Therefore,
\begin{equation}
    \mathcal P_{\infty}(m,n,r)
    =p_{x_c}(m)p_{x_h}(n)p_{x_c}(r).
    \label{eq:factorization}
\end{equation}

Thus $m$, $n$, and $r$ are independent random variables only in this complete-reset limit. The
work and hot-heat cumulants use the independent pair $(m,n)$, whereas cold-heat cumulants also
use the independent cold draw $r$.

\begin{figure*}[tbp]
    \centering
    \includegraphics[width=\linewidth]{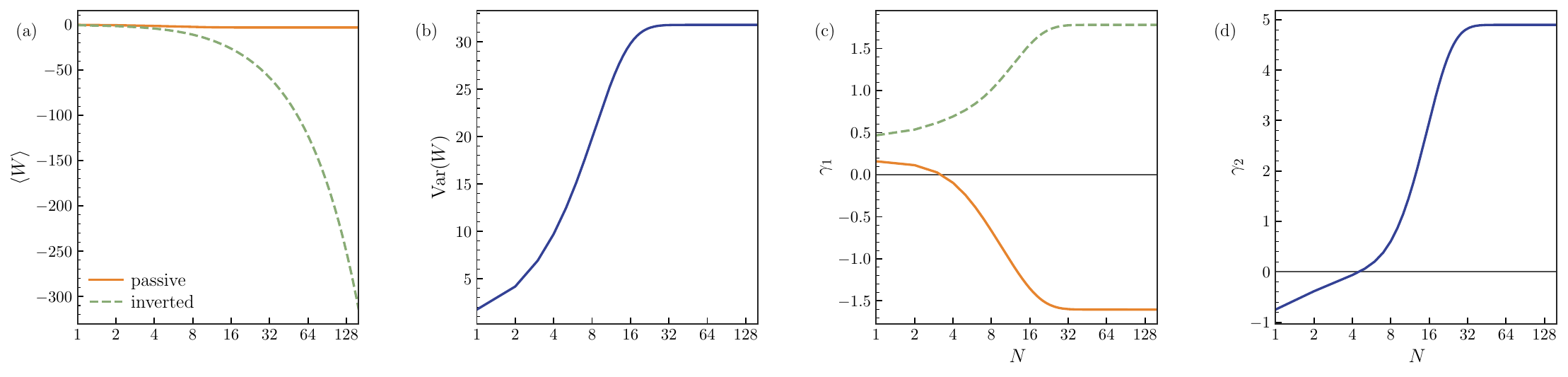}
\caption{Complete-reset work cumulants as functions of $N$ for matched passive and
inverted hot branches. The panels show (a) the mean work $\langle W\rangle$,
(b) the variance $\mathrm{Var}(W)$, (c) the standardized skewness $\gamma_1$,
and (d) the excess kurtosis $\gamma_2$. The variance and excess kurtosis coincide for
the matched passive and inverted branches and are therefore shown as single
curves in panels (b) and (d). Parameters are $\Omega_c=1$,
$\Omega_h=3$, $x_c=1$, and $|x_h|=0.375$. The curves follow
Eq.~\eqref{eq:cumulants-W} and display the passive--inverted parity structure
of the complete-reset work statistics.}
    \label{fig:work-cumulants}
\end{figure*}

\subsection{Characteristic function and closed-form cumulants}
\label{sec:fluct-cumulants}

In the complete-reset limit, the work performed in a single cycle is determined by two independent energy samples: the Dicke rung $m$ drawn from the cold endpoint and the rung $n$ drawn from the hot endpoint. Since the work strokes are deterministic for a given $m$ and $n$, sampling different pairs generates a distribution of work values,
\begin{equation}
P(W)=\sum_{m,n}p_{x_c}(m)p_{x_h}(n)\, \delta\!\left[W-\Delta\Omega(n-m)\right]. 
\label{eq:work-distribution-reset}
\end{equation}

This is the finite-ladder form of the TPM construction of work statistics. In
the present protocol the required energy basis is the Dicke basis at both ends of each work stroke.
Because the states are already diagonal in that basis and the work Hamiltonians commute, the
projective dephasing normally associated with a TPM construction does not alter the ensemble
dynamics. The energy labels therefore enter the distribution as classical population variables
within the reduced model.

The characteristic function of $P(W)$ \cite{Fei2022},
\begin{equation}
G(u)\equiv \left\langle e^{iuW}\right\rangle ,
\end{equation}
determines the full work statistics through Fourier transformation. Since $W=\Delta\Omega n-\Delta\Omega m$ is the difference of two independent endpoint contributions, this function factorizes into hot and cold components. Using Eq.~\eqref{eq:trajectory-WQ}, Eq.~\eqref{eq:factorization}, and $p_x(m)=e^{xm}/Z_j(x)$ yields:
\begin{equation}
    G(u)=
    \frac{Z_j(x_h+iu\Delta\Omega)}{Z_j(x_h)}
    \frac{Z_j(x_c-iu\Delta\Omega)}{Z_j(x_c)}.
    \label{eq:Gu-factored}
\end{equation}

Taking the logarithm of $G(u)$ separates the hot and cold endpoint fluctuations into independent work cumulants, bypassing repeated convolutions of the full work distribution. The first two cumulants yield the mean work and its variance, while higher orders quantify non-Gaussian features of the cycle-to-cycle statistics, such as asymmetry and excess tail weight. The $r$th work cumulant is
$\kappa_r(W)=i^{-r}\left.\partial_u^r\ln G(u)\right|_{u=0}$.

For the magnetization distribution, define
$\kappa_r^{(m)}(x)\equiv\partial_x^r\ln Z_j(x)$, with $\kappa_1^{(m)}(x)=\mu_N(x)$. The chain rule applied to Eq.~\eqref{eq:Gu-factored} then yields
\begin{equation}
    \kappa_r(W)=\Delta\Omega^r
    \left[\kappa_r^{(m)}(x_h)+(-1)^r\kappa_r^{(m)}(x_c)\right].
    \label{eq:cumulants-W}
\end{equation}
For $r=1$ this is the mean work already obtained in Eq.~\eqref{eq:means-complete-reset}; for $r=2$ it is the work variance. Since $Z_j(x)$ is even, $\kappa_r^{(m)}(x)$ has parity $(-1)^r$. For matched hot parameters $x_h=+a$ and $x_h=-a$, the even work cumulants are identical,
whereas the hot-state contribution to every odd cumulant changes sign. On the Dicke ladder, the
matched passive and inverted hot populations are mirror images, $p_{-a}(m)=p_{+a}(-m)$.
Population inversion moves the hot distribution to the opposite pole without changing its local
width, leaving the even hot-state contributions unchanged while reversing the odd hot-state
contributions. The two matched hot branches contribute equally to the work variance, even though their mean work can differ by order $N$.

Differentiating Eq.~\eqref{eq:muN} gives
\begin{equation}
\begin{aligned}
\kappa_2^{(m)}(x)=\partial_x\mu_N(x)
={}&\frac14\,\operatorname{csch}^2\!\left(\frac{x}{2}\right)\\
&-\frac{(N+1)^2}{4}\operatorname{csch}^2\!\left(\frac{(N+1)x}{2}\right).
\end{aligned}
\label{eq:kappa2m}
\end{equation}
and hence
\begin{equation}
    \operatorname{Var}(W)=\Delta\Omega^2
    \left[\kappa_2^{(m)}(x_h)+\kappa_2^{(m)}(x_c)\right].
    \label{eq:var-W}
\end{equation}

Higher cumulants follow by further differentiation. The mean, variance, standardized skewness, and
excess kurtosis are shown in Fig.~\ref{fig:work-cumulants}.

The finite-ladder TPM and cumulant construction applies the homothetic framework of Ref.~\cite{DamasHomothetic2026}. Complete reset decouples the endpoint samples, whereas finite contacts break this factorization and restore endpoint correlations (Sec.~\ref{sec:finite-fluct}).

\subsection{Reliability and the independent-copy benchmark}
\label{sec:fluct-reliability}

A large mean output alone does not guarantee a reproducible engine, as the useful work signal competes with intrinsic cycle-to-cycle fluctuations. This reliability is quantified by the signal-to-noise ratio
\begin{equation}
    \mathcal R_W\equiv\frac{|\langle W\rangle|}{\sqrt{\operatorname{Var}(W)}}
    =\frac{|\mu_N(x_h)-\mu_N(x_c)|}
    {\sqrt{\kappa_2^{(m)}(x_h)+\kappa_2^{(m)}(x_c)}}.
    \label{eq:RW-def}
\end{equation}

The units of energy cancel in this ratio. Experimentally, $\mathcal R_W$ compares the mean
extracted-work signal with its run-to-run standard deviation. Values much larger than unity
correspond to an output whose sign and magnitude are increasingly well resolved from one cycle
realization to the next. 

%Its reciprocal, $\phi_W\equiv\sqrt{\operatorname{Var}(W)}/|\langle W\rangle|$, is the relative work fluctuation; the two quantities carry the same information but are not numerically equal.

For fixed nonzero $x$, Eq.~\eqref{eq:kappa2m} approaches a constant as $N$ increases. The
variance therefore remains $O(1)$ on both matched branches. The mean is $O(1)$ on the passive branch
and $O(N)$ on the inverted branch, so
\begin{equation}
    \mathcal R_W=O(1)\quad(x_h>0),\qquad
    \mathcal R_W=O(N)\quad(x_h<0).
    \label{eq:RW-scaling}
\end{equation}

Geometrically, $\mathcal R_W$ in the complete-reset regime compares the separation between the cold and hot endpoint centers against their combined Dicke-rung distribution widths. On the passive branch, both distributions lie within $O(1)$ rungs of the same pole, restricting both separation and width to $O(1)$. On the inverted branch, the hot distribution reflects to the opposite pole; its width remains $O(1)$, but its center separates from the cold distribution by $O(N)$ rungs. The linear reliability growth is therefore driven entirely by this macroscopic signal increase, as the complete-reset work distribution does not broaden.

To distinguish this collective finite-ladder behavior from ordinary statistical self-averaging,
we compare it with $N$ independent one-constituent engines. The independent-copy benchmark is
obtained by adding the work of $N$ statistically independent one-constituent engines. Cumulants add, giving
\begin{equation}
\begin{aligned}
\langle W_N\rangle&=N\langle W_1\rangle,\\
\operatorname{Var}(W_N)&=N\operatorname{Var}(W_1),\\
\mathcal R_{W,\ind}(N)&=\sqrt N\,\mathcal R_{W,1}=O(\sqrt N).
\end{aligned}
\label{eq:RW-independent}
\end{equation}

In this limit, the inverted collective reliability exceeds the independent-copy benchmark by a factor $O(\sqrt N)$, whereas the passive collective branch remains $O(1)$ and is asymptotically below the $O(\sqrt N)$ independent benchmark, as shown in Fig.~\ref{fig:reliability}. Collective-engine studies have reported enhanced stability or reliability~\cite{souza2022collective,jaseem2023quadratic}; the matched Dicke benchmark used here separates the passive and inverted branches under the same finite-ladder statistics. This comparison fixes the constituent number and elementary one-constituent dynamics; it does not equate the accessible state-space dimension or the preparation resources of the symmetric collective architecture and $N$ independent engines.

\begin{figure}[tbp]
    \centering
    \includegraphics[width=0.8\linewidth]{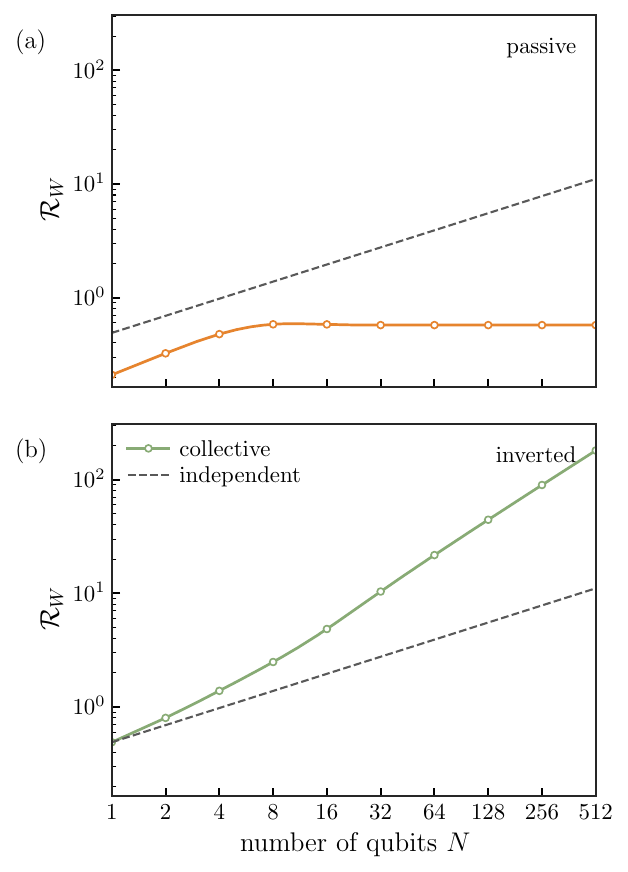}
    \caption{Complete-reset work reliability and independent-copy benchmark.
(a) Passive collective branch, for which the reliability saturates at
$\mathcal R_W=O(1)$, compared with the $O(\sqrt{N})$ reliability of
$N$ independent one-constituent engines. (b) Inverted collective branch,
whose reliability grows as $\mathcal R_W=O(N)$ and therefore exceeds the
independent-copy $O(\sqrt{N})$ scaling at large $N$. The collective
scalings reflect the Dicke-ladder geometry: inversion separates the hot and
cold endpoint distributions by $O(N)$ rungs while their local widths remain
$O(1)$. Curves are evaluated from Eqs.~\eqref{eq:RW-def} and
\eqref{eq:RW-independent} for $\Omega_c=1$, $\Omega_h=3$, $x_c=1$, and
$|x_h|=0.375$.}
    \label{fig:reliability}
\end{figure}

\section{Thermodynamic Resource Accounting}
\label{sec:cost}

The inverted branch can produce more gross work and larger reliability than the passive branch.
A resource-corrected comparison must also account for the nonequilibrium resource carried by the
inverted hot state. We therefore separate gross cycle output from reversible state-formation costs
before comparing the two operating modes. It is useful to distinguish two questions. The engine
calculation asks how much work is obtained from the supplied hot state. The resource calculation
asks how much reversible work is already embodied in preparing that state relative to a specified
passive or equilibrium reference. The second quantity is not heat exchanged during the Otto cycle
and is not a microscopic model of the pump that produces the inversion.

\subsection{Free energy and the cost boundary}
\label{sec:cost-setup}

For a state $\rho$ and a reference temperature $T>0$, the nonequilibrium free energy is
$F_T(\rho;H)\equiv\Tr(\rho H)-T S(\rho)$, with $S(\rho)\equiv-\Tr(\rho\ln\rho)$. For the Gibbs-Dicke state $\rho_x=e^{xJ_z}/Z_j(x)$ and the hot Hamiltonian
$H_h=-\Omega_hJ_z$, this becomes
\begin{equation}
\begin{aligned}
S_N(x)&=\ln Z_j(x)-x\mu_N(x),\\
F(x;\Omega_h,T)&=-\Omega_h\mu_N(x)-T S_N(x).
\end{aligned}
\label{eq:F-def}
\end{equation}

The reference equilibrium state at temperature $T$ has dimensionless parameter
$x_{\rm eq}=\Omega_h/T$. Its free-energy difference from $\rho_x$ is
\begin{equation}
    F(x;\Omega_h,T)-F(x_{\rm eq};\Omega_h,T)
    =T\,D\!\left(\rho_x\middle\|\rho_{x_{\rm eq}}\right)\geq0,
    \label{eq:free-energy-relative-entropy}
\end{equation}
where $D(\rho\|\sigma)=\Tr[\rho(\ln\rho-\ln\sigma)]$ is the quantum relative entropy. For
Dicke-diagonal states it reduces to the ordinary relative entropy of the population vectors.
Equation~\eqref{eq:free-energy-relative-entropy} gives the reversible state-formation lower bound
for the chosen Hamiltonian and reference temperature. For an experimental realization, this bound
should therefore be read as a best-case state-preparation cost. A real inversion protocol may
require more work because of finite-time control, dissipation, imperfect state preparation, and
reservoir reset. The same state function applies to an arbitrary Dicke-diagonal population vector,
so at finite contact time the resource comparison can be evaluated directly on the actual stationary
hot-corner populations rather than on the fully relaxed target Gibbs states. Device-level
preparation additionally depends on the physical pump, control, and reservoir-reset implementation.

We use three state-level quantities. The total formation bound is
$C_{\rm tot}(x;T)\equiv F(x;\Omega_h,T)-F(\Omega_h/T;\Omega_h,T)$. The passive total bound is obtained
by replacing $x$ by $+a$, and the excess inversion cost relative to the matched passive state is
$C_{\rm exc}(a)\equiv F(-a;\Omega_h,T)-F(+a;\Omega_h,T)$. The total bound depends on the equilibrium reference and thermodynamic boundary, whereas the excess
inversion cost is independent of $T$ for the matched symmetric pair. Physical reservoir preparation
and maintenance enter only in a device-level accounting beyond these state functions.

\subsection{Collective and product reference states}
\label{sec:cost-reference}

The dimension of the reference state is part of the cost definition. A \emph{collective-sector
reference} is normalized only on the selected $j=N/2$ ladder and asks how much reversible work is
needed within an already available symmetric architecture. A \emph{product reference} treats the
constituents as $N$ independent two-level systems and also charges correlations and support
restriction associated with the collective state. To make this distinction explicit, let
$y=\beta_{\rm ref}\Omega_h>0$ and define the one-constituent Gibbs state
\begin{equation}
    \tau_y=\frac{e^{y\sigma_z/2}}{2\cosh(y/2)},
    \qquad
    \tau_y^{\otimes N}
    =\frac{e^{yJ_z}}{[2\cosh(y/2)]^N}.
    \label{eq:tau-product}
\end{equation}

For an $N$-constituent state $\rho_N$ with one-body marginals $\rho_r$, the relative entropy to this product reference obeys
\begin{equation}
\begin{aligned}
D\!\left(\rho_N\middle\|\tau_y^{\otimes N}\right)
&=\sum_{r=1}^N D(\rho_r\|\tau_y)+I_N(\rho_N),\\
I_N(\rho_N)
&\equiv D\!\left(\rho_N\middle\|\bigotimes_{r=1}^N\rho_r\right)\\
&=\sum_{r=1}^N S(\rho_r)-S(\rho_N)\geq0.
\end{aligned}
\label{eq:corr-decomp}
\end{equation}

$I_N(\rho_N)$ is the total correlation (or multi-information): it measures the correlations contained in $\rho_N$ relative to the uncorrelated product of its one-body marginals, and vanishes if and only if $\rho_N=\bigotimes_{r=1}^N\rho_r$. For $N=2$, it reduces to the usual quantum mutual information.

For a permutation-symmetric state, all one-body marginals are equal and the first term is
$N D(\rho_1\|\tau_y)$. The second term is the total correlation stored in the $N$-body state.
The product and collective-sector references therefore represent different preparation boundaries.
The reference used for each cost comparison is stated explicitly.

The entropy derivative needed below follows from
$\mu_N(x)=\partial_x\ln Z_j(x)$:
\begin{equation}
    \frac{dS_N}{dx}=-x\,\kappa_2^{(m)}(x),
    \qquad
    \kappa_2^{(m)}(x)=\partial_x\mu_N(x)\geq0.
    \label{eq:SN-derivative}
\end{equation}

Hence $S_N(x)$ decreases for $x>0$ and increases for $x<0$.

\subsection{Ergotropy and the excess inversion cost}
\label{sec:cost-ergotropy}

Ergotropy is the maximum energy that can be extracted from a state through a cyclic unitary process, in which the Hamiltonian returns to its initial form:
\begin{equation}
    W_{\rm erg}(\rho,H)
    \equiv\Tr(\rho H)-\min_U\Tr(U\rho U^\dagger H).
    \label{eq:ergotropy-general}
\end{equation}

The minimizing state is the passive rearrangement: its eigenvalues are assigned in decreasing order
to the energy levels in increasing order. For the inverted Gibbs-Dicke state $\rho_{-a}$, the
passive rearrangement is $\rho_{+a}$ because $p_{-a}(-m)=p_{+a}(m)$ and
$E_h(m)=-\Omega_hm$. Since $S_N(-a)=S_N(+a)$,
\begin{equation}
\begin{split}
    W_{\rm erg}(a)
    &\equiv\Tr[H_h(\rho_{-a}-\rho_{+a})] \\
    &=2\Omega_h\mu_N(a).
\end{split}
\label{eq:Werg-def}
\end{equation}

The same expression is the excess free-energy difference:
\begin{equation}
    F(-a;\Omega_h,T)-F(+a;\Omega_h,T)=W_{\rm erg}(a).
    \label{eq:F-Werg-identity}
\end{equation}

Because the matched passive and inverted states have the same entropy, this ergotropy is exactly
their reversible excess free-energy difference at fixed $H_h$. We use $W_{\rm erg}$ as the excess
state-level cost of replacing the passive rearrangement by the inverted state. Experimentally, it
can be obtained from state tomography and a calibrated $H_h$, or from the energy change under a
unitary passive rearrangement. The total bound $C_{\rm tot}$ additionally requires a calibrated
reference temperature and the reconstructed-state entropy.

\subsection{Absolute and comparative balances}
\label{sec:cost-comparative}

For a complete-reset cycle with $x_h=-a$, the gross extracted work is
\begin{equation}
    W_{\rm gross,inv}(N)=-\langle W\rangle_{-a}
    =\Delta\Omega[\mu_N(x_c)+\mu_N(a)].
    \label{eq:Wgross-inv}
\end{equation}

Using the excess inversion cost as a state-level diagnostic gives
\begin{equation}
\begin{aligned}
\mathcal B_{\rm exc}(N)
&\equiv W_{\rm gross,inv}(N)-W_{\rm erg}(a)\\
&=\Delta\Omega\,\mu_N(x_c)-(\Omega_h+\Omega_c)\mu_N(a).
\end{aligned}
\label{eq:BN-def}
\end{equation}

The balance $\mathcal B_{\rm exc}$ compares the gross inverted output with the additional
state-level resource relative to the passive rearrangement.

For fixed positive $x_c$ and $a$, the large-$N$ expansion gives
\begin{equation}
\begin{aligned}
\mathcal B_{\rm exc}(N)={}&-\Omega_cN+(\Omega_h+\Omega_c)\bar n(a)\\
&-\Delta\Omega\,\bar n(x_c)+O(Ne^{-cN}).
\end{aligned}
\label{eq:BN-asymptotic}
\end{equation}

Here $\bar n(x)=1/(e^x-1)$ and one may take $c=\min(a,x_c)>0$. The negative linear term implies that this excess-cost
balance is negative for sufficiently large $N$, although it can be positive for small $N$. 

The negative large-$N$ behavior of $\mathcal B_{\rm exc}$ shows that gross inverted output is not a net advantage once the active-state resource is charged. As shown in Fig.~\ref{fig:cost-absolute}, the gross inverted output continues to grow with $N$, while the excess-cost balance decreases and becomes negative, separating gross output from resource-corrected performance. Population inversion can nevertheless remain operationally useful when the active resource is persistent, externally supplied, or when power and reliability are prioritized over net energetic gain. In those cases, the preparation, maintenance, and refresh mechanisms must be included in the thermodynamic boundary.

\begin{figure}[tbp]
    \centering
    \includegraphics[width=0.85\linewidth]{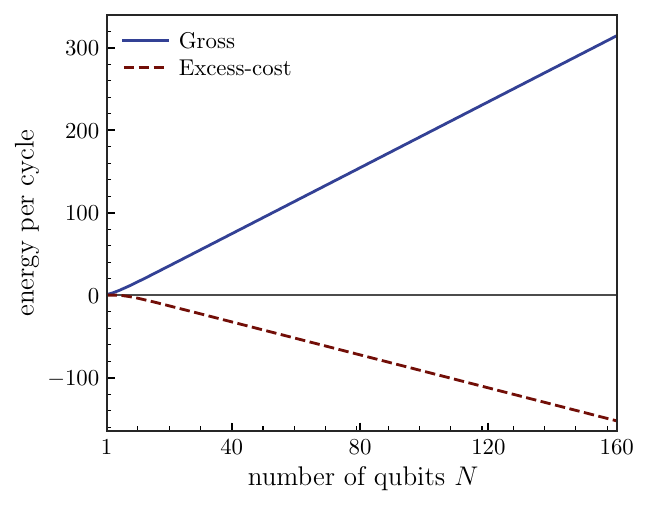}
\caption{Absolute resource balance for the inverted branch.
The solid curve shows the complete-reset gross extracted work
$W_{\rm gross,inv}$, while the dashed curve shows the excess-cost balance
$\mathcal B_{\rm exc}=W_{\rm gross,inv}-W_{\rm erg}$ defined in
Eq.~\eqref{eq:BN-def}. The gross output grows with $N$, but the
resource-corrected balance decreases and becomes negative because the
ergotropy cost of the inverted state exceeds the additional output.
Parameters are $\Omega_c=1$, $\Omega_h=3$, $x_c=1$, and $a=0.375$.}
    \label{fig:cost-absolute}
\end{figure}

The more stringent comparison keeps the passive alternative available. Consider two complete-reset
cycles with the same cold state $\rho_{x_c}$ and matched hot parameters $+a$ and $-a$. The gross
work gain from replacing the passive hot state by the inverted one is
\begin{equation}
    \Delta W_{\rm gross}
    =2\Delta\Omega\,\mu_N(a)
    =\frac{\Delta\Omega}{\Omega_h}W_{\rm erg}(a)
    =\eta W_{\rm erg}(a).
    \label{eq:matched-gross}
\end{equation}

Equation~\eqref{eq:matched-gross} gives a direct physical interpretation of the inverted gain:
the additional Otto work is a fraction $\eta$ of the unitary-extractable energy already stored in
the inverted hot state. The collective scaling can amplify the size of this resource with $N$, but
it does not change the conversion fraction set by the Otto geometry.

Charging the excess cost once per cycle gives
\begin{equation}
    \Delta W_{\rm net}
    =-(1-\eta)W_{\rm erg}(a)
    =-\frac{\Omega_c}{\Omega_h}W_{\rm erg}(a)<0.
    \label{eq:matched-net}
\end{equation}

This conclusion is independent of $N$ and is stronger than the absolute balance: even if the
inverted cycle covers its excess cost relative to a zero-output reference, it does not outperform
the already available matched passive cycle.

The same conclusion holds for any passive comparator $\rho_b$ that is itself an engine with the
same cold state, namely $0<b<x_c$. Its gross gain is
$\Delta W_{\rm gross}=\Omega_h\eta[\mu_N(a)+\mu_N(b)]$, while its excess state-level cost is
\begin{equation}
\begin{aligned}
\Delta C={}&F(-a;\Omega_h,T_c)-F(b;\Omega_h,T_c)\\
={}&\Omega_h[\mu_N(a)+\mu_N(b)]-T_c[S_N(a)-S_N(b)].
\end{aligned}
\label{eq:general-cost}
\end{equation}

Here $T_c=1/\beta_c$ and $T_cx_c=\Omega_c$. Thus
\begin{equation}
    \Delta W_{\rm net}
    =T_c[S_N(a)-S_N(b)]-\Omega_c[\mu_N(a)+\mu_N(b)].
    \label{eq:general-net}
\end{equation}

If $a\geq b$, both terms on the right-hand side are nonpositive. If $0<a<b<x_c$,
Eq.~\eqref{eq:SN-derivative} gives
\begin{equation}
\begin{aligned}
S_N(a)-S_N(b)&=\int_a^b t\,\kappa_2^{(m)}(t)\,dt\\
&<x_c[\mu_N(b)-\mu_N(a)]\\
&<x_c[\mu_N(a)+\mu_N(b)].
\end{aligned}
\label{eq:entropy-bound}
\end{equation}

Multiplying by $T_c$ yields $T_c[S_N(a)-S_N(b)]<\Omega_c[\mu_N(a)+\mu_N(b)]$, and therefore $\Delta W_{\rm net}<0$. The entropy contribution can reduce the formation-cost difference but cannot compensate the polarization term required for the inverted cycle to outperform an already operating passive engine within this reference boundary. As shown in Fig.~\ref{fig:cost-comparative}, this conclusion holds both for the matched passive comparison and for the broader class of passive comparators $0<b<x_c$. The result applies to complete reset, the symmetric Dicke reference, and the reversible state-level cost difference.

\begin{figure}[tbp]
    \centering
    \includegraphics[width=0.85\linewidth]{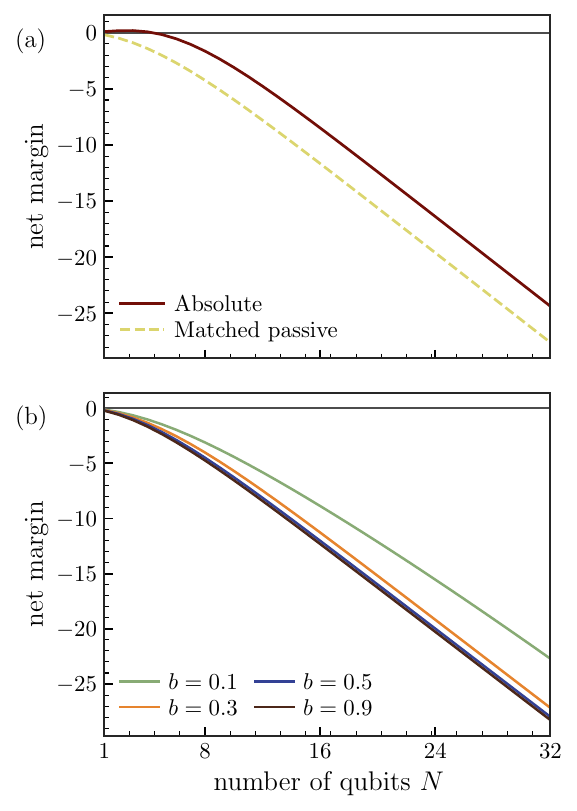}
\caption{Resource-corrected comparison of passive and inverted operation.
    Panel (a) contrasts the absolute excess-cost balance with the incremental
    net gain relative to the matched passive engine. Panel (b) shows the
    passive-to-inverted net margin for passive comparators $0<b<x_c$,
    evaluated from Eq.~\eqref{eq:general-net}. In both comparisons,
    charging the reversible state-level resource removes the gross inverted
    advantage over an operating passive engine. Parameters are
    $\Omega_c=1$, $\Omega_h=3$, $x_c=1$, and $a=0.375$.}
    \label{fig:cost-comparative}
\end{figure}

\subsection{Conditional amortization benchmark}
\label{sec:cost-reuse}

The comparison above charges the reversible excess state-formation resource once per cycle. To
examine a different accounting boundary, suppose instead that one external preparation charge
$W_{\rm erg}(a)$ is allocated across an integer number $\ell$ of otherwise identical cycles, while
the active hot state required by each cycle is re-established without any additional charged excess
resource during that interval. Under this shared-cost assumption, the amortized excess balance per
cycle is
\begin{equation}
    \Delta W_{\rm net}^{(\ell)}
    =\eta W_{\rm erg}(a)-\frac{W_{\rm erg}(a)}{\ell}
    =\left(\eta-\frac1\ell\right)W_{\rm erg}(a).
    \label{eq:reuse-gain}
\end{equation}

Its sign changes when
\begin{equation}
    \ell>\ell^\star,\qquad
    \ell^\star\equiv\frac1\eta=\frac{\Omega_h}{\Omega_h-\Omega_c}.
    \label{eq:reuse-threshold}
\end{equation}

Because $\ell$ is an integer, the first positive balance occurs at the smallest integer strictly
larger than $\ell^\star$. The threshold follows from the stipulated sharing of one preparation
charge across several cycles. A physical realization of this accounting boundary requires a
persistent reservoir, battery, catalytic resource, or pump whose degradation, refresh, maintenance,
and control costs determine whether the threshold survives at device level.

\section{Finite-Contact Fluctuations and Approximation Hierarchy}
\label{sec:finite-fluct}

The complete-reset formulas in Sec.~\ref{sec:fluctuations} are not the formulas used when the
contacts are finite. A finite contact retains information about the incoming Dicke rung, producing
correlations across the corners of one cycle and a dependence of the Markovian propagator output on
the initial state. This memory reflects incomplete relaxation within the Markovian population
dynamics rather than non-Markovian reservoir memory. This section gives the exact finite-state
construction and then states separately what is approximated in the large-$N$ analysis.

\subsection{Exact finite-contact path probabilities}
\label{sec:finite-path}

Let $m_1$, $m_3$, and $m_5$ be the Dicke labels at corners 1, 3, and 5, respectively. A complete
one-cycle path has probability
\begin{equation}
    \mathcal P(m_1,m_3,m_5)
    =p_1(m_1)[K_h]_{m_3,m_1}[K_c]_{m_5,m_3}.
    \label{eq:three-corner-path}
\end{equation}

A single stochastic realization does not necessarily return the working medium to its initial energy. The cycle closes only statistically: stationarity ensures identical distributions at corners 1 and 5, even though $m_5\neq m_1$ on individual trajectories. The trajectory-level imbalance $W+Q_h+Q_c$ represents the fluctuating change in the working-medium internal energy, whose ensemble average vanishes in the stationary cycle.

At the stationary cycle, the marginal distribution of $m_5$ equals that of $m_1$, but the two
labels are not independent for a finite contact. The work and heat assigned to this path are
\begin{equation}
\begin{aligned}
W&=\Delta\Omega(m_3-m_1),\\
Q_h&=-\Omega_h(m_3-m_1),\\
Q_c&=\Omega_c(m_3-m_5).
\end{aligned}
\label{eq:path-work-heat}
\end{equation}

The work and hot heat are tightly coupled on each path, but the cold heat includes the final
corner label. The pathwise first law is
\begin{equation}
    W+Q_h+Q_c=\Omega_c(m_1-m_5)
    =E_c(m_5)-E_c(m_1).
    \label{eq:path-first-law}
\end{equation}

Only its stationary average vanishes, because $\langle m_5\rangle=\langle m_1\rangle$.

The exact one-cycle moments can be obtained without sampling. For any function $f$,
\begin{equation}
    \langle f\rangle_{\rm cyc}
    =\sum_{m_1,m_3,m_5}\mathcal P(m_1,m_3,m_5)
    f(m_1,m_3,m_5).
    \label{eq:path-moment}
\end{equation}

In particular, the work variance can be written in terms of endpoint moments as
\begin{equation}
    \operatorname{Var}(W)=\Delta\Omega^2
    \left[\operatorname{Var}(m_3)+\operatorname{Var}(m_1)
    -2\operatorname{Cov}(m_1,m_3)\right],
    \label{eq:finite-var-decomp}
\end{equation}
where
\begin{equation}
    \operatorname{Cov}(m_1,m_3)
    =\sum_{m_1,m_3}p_1(m_1)[K_h]_{m_3,m_1}m_1m_3-M_1M_3.
    \label{eq:finite-covariance}
\end{equation}

Equation~\eqref{eq:finite-var-decomp} separates two sources of work noise. The first two terms capture the widths of the cold and hot endpoint distributions, while the covariance quantifies the hot corner's memory of the initial rung. Complete reset erases this memory, rendering $m_3$ and $m_1$ statistically independent; the covariance vanishes, recovering Eq.~\eqref{eq:var-W}.

For finite contacts, the conditional distribution $[K_h]_{m_3,m_1}$ retains information about $m_1$. A positive covariance indicates that initial fluctuations from the mean polarization persist through the hot contact. Since work depends on the difference $m_3-m_1$, this correlated fluctuation partially cancels, reducing the work variance, whereas a negative covariance amplifies it. The covariance directly quantifies the memory of the incoming polarization after an incomplete hot contact.

Because the work strokes commute with $J_z$ and the states remain Dicke-diagonal, these trajectory
labels act as classical stochastic energy records within the reduced model, and the path
distribution corresponds directly to repeated energy-resolved cycle realizations. In noncommuting
protocols a TPM construction generally dephases the state~\cite{denzler2020efficiency,Aguilar2026};
here the states are already Dicke-diagonal and the work strokes preserve that basis.

For fixed $N$, rates, and contact durations, the finite-state trajectory statistics are exact within the collective Davies-form model. Path-sum convergence is verified by systematically lowering the probability cutoff.

\subsection{Tilted maps for joint finite-contact statistics}
\label{sec:tilted-maps}

The explicit path sum of Sec.~\ref{sec:finite-path} gives a direct description of the stochastic
thermodynamics of one cycle. It becomes less convenient, however, when one wants higher
cumulants, correlations between work and heat, or statistics accumulated over many cycles.
Full counting statistics provides a compact way of performing the same sum through a
generating function
~\cite{esposito2009nonequilibrium,Landi2024FCS}. The basic idea is simple. For a stochastic
quantity $X$, the moment-generating function
\begin{equation}
    G_X(s)=\left\langle e^{sX}\right\rangle
\end{equation}
contains its statistical moments, while derivatives of $\ln G_X(s)$ at $s=0$ give its cumulants.
For several stochastic quantities, separate counting fields are introduced and mixed derivatives
generate their correlations. This is the real-field version of the characteristic-function
construction used in Sec.~\ref{sec:fluct-cumulants}, where the corresponding weight was
$e^{iuX}$.

For a Markov process, the generating function can be constructed without listing every trajectory
explicitly. Suppose that a transition $m\to m'$ during contact $\alpha$ changes the
working-medium energy by
$E_\alpha(m')-E_\alpha(m)$. Multiplying that transition rate by
\[
   e^{\left[ s_\alpha\bigl(E_\alpha(m')-E_\alpha(m)\bigr) \right]}
\]
assigns the appropriate statistical weight to the jump. Along a trajectory, these factors
multiply, so their product is the exponential of the total heat accumulated on that contact.
The matrix exponential of the resulting generator then sums these weights over all possible
jump histories automatically. This reweighting of the physical generator is what is meant by
a \emph{tilted} generator.

For the present birth--death process, we therefore define
\begin{equation}
\begin{split}
    [R_\alpha(s_\alpha)]_{m',m} ={} [R_\alpha]_{m',m}  \times e^{\big[s_\alpha \big(E_\alpha(m')-E_\alpha(m)\big)\big]}, \\
    m'\ne m.
\end{split}
\label{eq:tilted-generator}
\end{equation}

The diagonal escape rates are left unchanged. The corresponding tilted contact propagator is
\begin{equation}
    K_\alpha(s_\alpha)
    =e^{[R_\alpha(s_\alpha)\tau_\alpha]}.
\end{equation}

At zero counting field, $R_\alpha(0)=R_\alpha$ and
$K_\alpha(0)=K_\alpha$, so the ordinary population dynamics is recovered. For
$s_\alpha\neq0$, the tilted generator is not a stochastic generator describing a modified
reservoir; it is a statistical bookkeeping device whose matrix elements carry the weights needed
to reconstruct the heat distribution.

The work strokes are even simpler because the Dicke label does not change during them. For a
given label $m$, compression contributes the deterministic work $-\Delta\Omega m$, whereas
expansion contributes $+\Delta\Omega m$. Their counting factors can therefore be represented by
the diagonal matrices
\begin{equation}
\begin{aligned}
D_{c\to h}(s_W)&=\diag_m \, e^{(-s_W\Delta\Omega m)},\\
D_{h\to c}(s_W)&=\diag_m \, e^{(+s_W\Delta\Omega m)}.
\end{aligned}
\label{eq:tilted-work}
\end{equation}

Combining the four strokes gives the joint moment-generating map for one cycle,
\begin{equation}
\begin{aligned}
\mathcal M(\boldsymbol s)={}&
K_c(s_c)D_{h\to c}(s_W)
K_h(s_h)D_{c\to h}(s_W),\\
\boldsymbol s={}&(s_W,s_h,s_c).
\end{aligned}
\label{eq:tilted-cycle}
\end{equation}

The ordering follows the physical protocol from right to left: compression is followed by the hot
contact, expansion, and finally the cold contact. Each factor therefore performs the ordinary
population propagation while attaching the appropriate statistical weight to the corresponding
work or heat increment.

Starting from the stationary population $\vec p_1$, the joint generating function for a single
cycle is
\begin{equation}
    G_1(\boldsymbol s)
    =\boldsymbol 1^{\mathsf T}
    \mathcal M(\boldsymbol s)\vec p_1 .
    \label{eq:G1}
\end{equation}

The multiplication by $\boldsymbol 1^{\mathsf T}$ simply sums over the final Dicke label, because
the final state itself is not conditioned upon when constructing the work--heat distribution.
At zero counting fields,
\begin{equation}
    \mathcal M(\boldsymbol 0)=K_cK_h,
    \qquad
    G_1(\boldsymbol 0)=1,
\end{equation}
as required by normalization.

The connection with measurable fluctuations is then immediate. For example,
\begin{equation}
\begin{aligned}
\langle W\rangle
&=\left.\frac{\partial}{\partial s_W}
        \ln G_1(\boldsymbol s)\right|_{\boldsymbol s=0},\\
\operatorname{Var}(W)
&=\left.\frac{\partial^2}{\partial s_W^2}
        \ln G_1(\boldsymbol s)\right|_{\boldsymbol s=0},\\
\operatorname{Cov}(W,Q_h)
&=\left.
\frac{\partial^2}{\partial s_W\,\partial s_h}
\ln G_1(\boldsymbol s)\right|_{\boldsymbol s=0}.
\end{aligned}
\label{eq:tilted-cumulants-example}
\end{equation}

The second derivatives have their standard statistical meaning: $\operatorname{Var}(W)=\langle W^2\rangle-\langle W\rangle^2$ is the spread of the work distribution, while $\operatorname{Cov}(W,Q_h)=\langle WQ_h\rangle-\langle W\rangle\langle Q_h\rangle$ is the correlation between work and hot-heat fluctuations.

Higher derivatives generate higher cumulants and mixed work--heat correlations. Thus the tilted-map
construction contains exactly the same trajectory statistics as the explicit positive path
probabilities of Sec.~\ref{sec:finite-path}; its advantage is that the sum over trajectories is
performed algebraically by matrix propagation. This becomes particularly useful when the cycle is
repeated, as discussed next.

\subsection{Repeated cycles and temporal correlations}
\label{sec:repeated-fluct}

The main advantage of the tilted-map representation becomes apparent when the engine is operated
repeatedly. Instead of enumerating trajectories over many cycles, one simply composes the same
counting-field-dependent cycle map. This composition retains the stochastic boundary state passed
from one cycle to the next and therefore includes temporal correlations automatically.

A stationary engine does not guarantee statistically independent cycles. Experimentally, this
distinction would appear in a time series of cycle-resolved work values: even after the average
output has reached a stationary value, a high- or low-polarization outcome at the end of one cycle
can bias the work observed in the next.

For $K$ consecutive stationary cycles, the joint generating function is
\begin{equation}
 G_K(\boldsymbol s)
 =\boldsymbol 1^{\mathsf T}\mathcal M(\boldsymbol s)^K\vec p_1.
 \label{eq:GK}
\end{equation}

The long-time scaled cumulant-generating function per cycle is
$\psi(\boldsymbol s)=\ln\Lambda_0(\boldsymbol s)$, where $\Lambda_0$ is the eigenvalue of largest modulus and
$\Lambda_0(\boldsymbol 0)=1$. The one-cycle cumulants and the long-time cumulants are different
observables unless the contacts erase the boundary memory. Iterating the tilted cycle map progressively isolates the dominant eigenmode, whose eigenvalue governs the extensive growth of the accumulated work and heat statistics. The subdominant modes dictate boundary-state relaxation and mediate inter-cycle correlations.

Consider a single work counting parameter and write
\begin{equation}
\begin{split}
 \mathcal M(s)=M_0+sM_1+\frac{s^2}{2}M_2+O(s^3),\\
 \qquad M_0=K_cK_h.
\end{split}
\label{eq:tilted-expansion}
\end{equation}

The coefficient $\lambda_1=\boldsymbol 1^{\mathsf T}M_1\vec p_1$ is the stationary mean counted work per cycle. The normalized Poisson equation
\begin{equation}
 (\id-M_0)\boldsymbol u
 =(M_1-\lambda_1\id)\vec p_1,
 \qquad
 \boldsymbol 1^{\mathsf T}\boldsymbol u=0
 \label{eq:poisson}
\end{equation}
is then solved. The vector $\boldsymbol u$ describes the first-order change of the cycle-boundary distribution induced by the counting field. It records how fluctuations generated in one cycle modify the distribution entering later cycles, and supplies the inter-cycle contribution to the long-time variance.

The second derivative of the dominant eigenvalue is
\begin{equation}
 \lambda_2=\boldsymbol 1^{\mathsf T}M_2\vec p_1
 +2\boldsymbol 1^{\mathsf T}M_1\boldsymbol u,
 \qquad
 c_{2,\infty}=\lambda_2-\lambda_1^2.
 \label{eq:longvar}
\end{equation}

Equivalently, the long-time variance per cycle is the covariance sum
\begin{equation}
 c_{2,\infty}
 =\operatorname{Var}(W_0)
 +2\sum_{q=1}^{\infty}\operatorname{Cov}(W_0,W_q).
 \label{eq:covsum}
\end{equation}

Equation~\eqref{eq:covsum} decomposes the long-time variance into single-cycle fluctuations and an inter-cycle covariance sum. Positive correlations amplify the long-time variance per cycle, whereas negative correlations induce partial cancellation, stabilizing the output. Complete erasure of the cycle-boundary memory destroys these covariances, equating the long-time and single-cycle fluctuations.

The corresponding one-cycle and long-time reliabilities are
\begin{equation}
 \mathcal R_{W,1}=\frac{|\avg W|}{\sqrt{\operatorname{Var}(W_0)}},
 \qquad
 \mathcal R_{W,\infty}=\frac{|\avg W|}{\sqrt{c_{2,\infty}}}.
 \label{eq:two-reliabilities}
\end{equation}

Positive inter-cycle correlations degrade long-time stability despite favorable single-cycle reliability, whereas anticorrelations suppress accumulated noise. The reliability of a continuously operated engine fundamentally depends on these memory effects rather than single-cycle statistics alone.

Derivatives of the matrix exponentials are evaluated with Fr\'echet derivatives or equivalent
block-exponential formulas rather than finite differences of the counting field. Equations
~\eqref{eq:tilted-expansion}--\eqref{eq:covsum} then give the finite-$N$ repeated-cycle cumulants
directly from the stationary cycle map.

Counting-field methods capture joint work--heat statistics in quantum Otto cycles~\cite{esposito2009nonequilibrium,campisi2011colloquium,
mohanta2023fullstats}, while temporal correlations govern the long-time stability of finite-time engines~\cite{xuwatanabe2022correlation}. Applied to the collective Dicke cycle, this framework reveals how the shared stochastic boundary state explicitly separates single-cycle and long-time fluctuations.

\subsection{Finite-time numerical observables and resource balance}
\label{sec:finite-numerics}

Finite-contact observables are evaluated from matrix exponentials for $K_h$ and $K_c$, the stationary
vector of $K_cK_h$, and the positive path sum in Eq.~\eqref{eq:path-moment}. The controlled reference
point uses $\Omega_c=1$, $\Omega_h=3$, $x_c=1$, $x_h=-0.375$,
$\gamma=4\times10^{-5}$, $\tau_h=\tau_c=2500$, and zero-duration commuting work strokes, so the
dimensionless exposure is $\gamma\tau_h=\gamma\tau_c=0.1$. This rate--time rescaling leaves the
contact propagators and stationary populations unchanged while placing the collective rates inside
the secular-control criterion of Sec.~\ref{sec:validity}. In physical terms, the elementary
coupling is weakened while the contact is lengthened by the same factor. The system therefore
experiences the same dimensionless relaxation exposure, but the instantaneous collective rates are
kept safely below the system frequencies used in the secular description.

Accordingly, the relative work fluctuation is unchanged by this rescaling: it is $13.64$ at
$N=1$ and $2.79$ at $N=8$, while the corresponding reliability is $0.073$ and $0.359$. Absolute powers scale with $\gamma$; dimensionless ratios and local finite-size
exponents are unchanged. We treat these as finite-$N$ finite-contact results and do not infer an
asymptotic inverted-branch power exponent from this scan.
The $O(1)$ and $O(N)$ exponents derived in Sec.~\ref{sec:means-scaling} refer to complete-reset
work per cycle; the analytical finite-contact power limit derived above is the passive result in
Eq.~\eqref{eq:power-infty-pass}.

To separate a genuinely collective finite-contact enhancement from the trivial extensivity obtained
by operating more constituents, we compare the inverted collective engine with $N$ independent
inverted one-qubit engines under the same matched-total-rate convention and the same exposure
$u$. At the reference point the two descriptions coincide for $N=1$, while the ratio
$P_{\rm coll}^{\rm inv}/P_{\rm ind}^{\rm inv}$ grows monotonically over the tested sizes and reaches
$5.96$ at $N=32$ [Fig.~\ref{fig:finite-contact-bridge}(a)]. This ratio is invariant under the common rate--time rescaling and isolates the collective Dicke
kinetics within the stated matched-total-rate convention. It is therefore a kinetic benchmark
rather than an equal-device-resource comparison.

For finite contacts, the hot corner is generally not the fully relaxed state $\rho_{\pm a}$. We
therefore evaluate the same cold-reference free energy directly on the actual stationary hot-corner
populations $p_{3,\rm inv}$ and $p_{3,\rm pass}$. The two costs below answer different
bookkeeping questions. $C_{\rm exc}^{\rm fc}$ asks for the additional state resource carried by
the inverted hot corner relative to the passive hot corner, whereas $C_{\rm tot,inv}^{\rm fc}$
asks for its total nonequilibrium distance from the chosen cold-temperature equilibrium reference.
With
$\rho_{\rm eq,c}\propto\exp[(\Omega_h/T_c)J_z]$ on the symmetric Dicke ladder, define
\begin{equation}
\begin{aligned}
C_{\rm exc}^{\rm fc}
&=F_{T_c}(p_{3,\rm inv};H_h)-F_{T_c}(p_{3,\rm pass};H_h),\\
C_{\rm tot,inv}^{\rm fc}
&=F_{T_c}(p_{3,\rm inv};H_h)-F_{T_c}(\rho_{\rm eq,c};H_h).
\end{aligned}
\label{eq:finite-contact-corner-costs}
\end{equation}

The corresponding power differences are
\begin{equation}
\begin{aligned}
\Delta P_{\rm exc}^{\rm fc}
&=P_{\rm inv}-P_{\rm pass}-\frac{C_{\rm exc}^{\rm fc}}{\tau_{\rm cyc}},\\
\Delta P_{\rm tot}^{\rm fc}
&=P_{\rm inv}-P_{\rm pass}-\frac{C_{\rm tot,inv}^{\rm fc}}{\tau_{\rm cyc}}.
\end{aligned}
\label{eq:finite-contact-corner-margins}
\end{equation}

The first line compares the actual inverted and passive hot corners on the same reversible
state-level reference. The second is a stricter benchmark that charges the inverted hot corner its
full formation bound relative to $\rho_{\rm eq,c}$ while leaving the passive engine as an available
gross-output comparator; it is not a symmetric net-to-net comparison. Both quantities are state-level accounting benchmarks assigned per cycle. They do not represent the
microscopic work supplied by the reservoir or pump during a finite contact, which requires an
explicit device-level model.

\begin{figure}[tbp]
    \centering
    \includegraphics[width=0.85\linewidth]{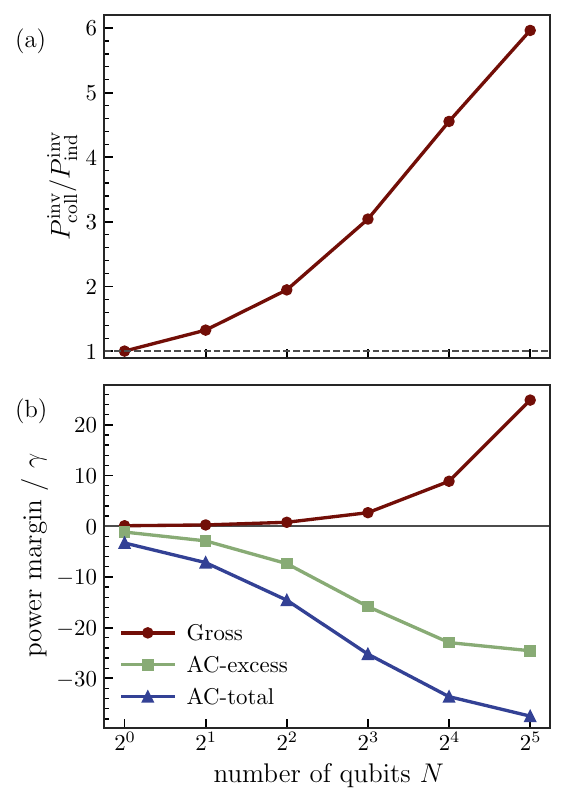}
    \caption{Finite-contact separation of collective kinetics and state-level resource accounting
    under the matched-total-rate convention. (a) Ratio of the inverted collective gross power
    to the gross power of $N$ independent inverted one-qubit engines at the same exposure $u$; the
    horizontal line marks equality. (b) Gross inversion gain and the two actual-corner (AC)
state-level corrected quantities, divided by $\gamma$.
    The excess-corrected curve compares the actual inverted and passive hot corners; the total
    benchmark charges the inverted hot corner relative to the cold-temperature equilibrium
    reference. The controlled reference has $\Omega_c=1$, $\Omega_h=3$, $x_c=1$,
    $x_h=-0.375$, $\gamma=4\times10^{-5}$, and $\tau_h=\tau_c=2500$, with zero-duration
    commuting work strokes.}
    \label{fig:finite-contact-bridge}
\end{figure}

At the controlled reference point the gross inversion gain is positive for every tested
$N=1,2,4,8,16,32$, while both actual-corner corrected quantities in
Eq.~\eqref{eq:finite-contact-corner-margins} remain negative
[Fig.~\ref{fig:finite-contact-bridge}(b)]. Because every rate-like quantity rescales with the common
elementary rate, panel (b) reports the margins in units of $\gamma$. At $N=32$ the rate-scaled gross
gain, actual-corner excess margin, and actual-corner total benchmark are $24.869$, $-24.563$, and
$-37.452$, respectively.

This sign pattern persists across the validation grid used in the reproducibility notebooks:
$N\in\{2,8,16\}$, five inverted hot parameters spanning $|x_h|=0.09$--$1.8$, and five equal
exposures spanning $u_h=u_c=0.005$--$0.8$, for a total of $75$ points. At
$\gamma=4\times10^{-5}$ these correspond to contact times $125$--$20000$. The gross inversion gain
is positive at all $75$ points, while neither actual-corner corrected quantity becomes positive on
that grid. The least-negative values are $-3.14\times10^{-6}$ and $-3.59\times10^{-5}$ in absolute
power units for the excess and total benchmarks, respectively. This is a finite-contact numerical result for the stated parameter domain and does not establish
the corresponding statement for arbitrary contact parameters.

\subsection{Large-\texorpdfstring{$N$}{N} approximations and their domains}
\label{sec:largeN-hierarchy}

Two distinct large-$N$ approximations are used as diagnostics. The first is a Holstein--Primakoff
approximation near a pole. With $k=j-m$ and $k=O(1)$, it replaces the collective ladder by a
bosonic occupation variable and describes a corner distribution concentrated near $m=j$ for
$x>0$. It is a local approximation and is not a description of the finite-time passage from
$m=j$ to $m=-j$ on the inverted branch. For the relaxed-endpoint diagnostic used below, the
Holstein--Primakoff approximation is implemented through the leading north- and south-pole forms
already displayed in Eq.~\eqref{eq:muN-asymptotic}, applied separately to the cold and
inverted-hot endpoints.

The second is a finite-time Kramers--Moyal approximation \cite{vankampen2007}. It expands the discrete birth--death
process in a continuous coordinate and retains the first two jump moments. We introduce the
rescaled coordinate
\begin{equation}
    z\equiv\frac{m}{j},
    \qquad
    \delta\equiv\frac{1}{j}=\frac{2}{N},
    \label{eq:km-coordinate}
\end{equation}
where $z=+1$ and $z=-1$ are the north and south poles of the Dicke ladder, respectively. For a
given contact, we suppress the contact index $\alpha$ on the elementary rates. The exact
nearest-neighbor rates expressed in this coordinate are
\begin{equation}
\begin{aligned}
w_+(z)
&=j(1-z)\big[j(1+z)+1\big]\Gamma_{\downarrow},\\
w_-(z)
&=j(1+z)\big[j(1-z)+1\big]\Gamma_{\uparrow},
\end{aligned}
\label{eq:km-rates}
\end{equation}
where $w_+$ and $w_-$ correspond to steps $z\to z+\delta$ and
$z\to z-\delta$, respectively. The associated Kramers--Moyal drift and diffusion coefficients are
\begin{equation}
\begin{aligned}
A(z)&=\delta\,[w_+(z)-w_-(z)],\\
B(z)&=\delta^2[w_+(z)+w_-(z)].
\end{aligned}
\label{eq:km-coefficients}
\end{equation}

The linear-noise approximation used in the finite-time comparison propagates the mean
$\bar z$ and variance $V=\operatorname{Var}(z)$ according to
\begin{equation}
\begin{aligned}
\frac{d\bar z}{dt}&=A(\bar z),\\
\frac{dV}{dt}&=2A'(\bar z)V+B(\bar z).
\end{aligned}
\label{eq:km-linear-noise}
\end{equation}

The drift therefore estimates the mean polarization, while the diffusion term determines the
leading fluctuation correction. This approximation is distinct from the Holstein--Primakoff
representation and from the exact finite-dimensional propagator. The Holstein--Primakoff
construction and the system-size expansion of a master equation are standard approximations, but
their validity is local in state space and parameter dependent
~\cite{holstein1940,vankampen2007}.

The Holstein--Primakoff error becomes small for the relaxed endpoint test, in which separate local expansions are made about the cold north pole and the inverted hot south pole. The convergence of this endpoint approximation is shown in Fig.~\ref{fig:largeN-validation}: the relative Holstein--Primakoff work error falls below the one-percent level within the tested range of $N$. The Kramers--Moyal comparison instead tests the finite-time evolution of the rescaled first and second moments. In Table~\ref{tab:largeN-errors}, the Holstein--Primakoff column reports the relative error in the relaxed endpoint work, whereas the Kramers--Moyal columns report the maximum absolute errors in the rescaled corner mean and variance over the two stationary-cycle corners. The Kramers--Moyal variance error is not monotonic. The two approximations probe different limits and are kept separate from the exact finite-state inverted-cycle calculation.

\begin{figure}[tbp]
    \centering
    \includegraphics[width=0.9\linewidth]{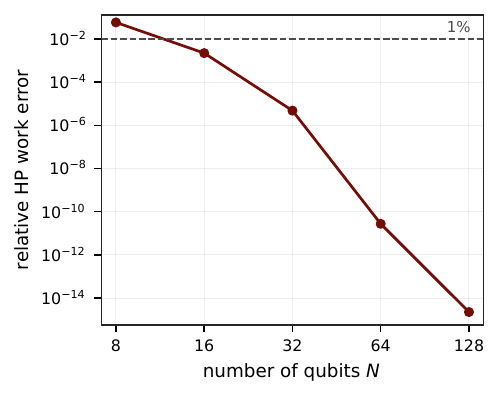}
    \caption{Local Holstein--Primakoff endpoint approximation compared with the exact relaxed
    inverted-cycle work. Separate pole-local expansions are used about the cold north pole and the
    inverted hot south pole for $\Omega_c=1$, $\Omega_h=3$, $x_c=1$, and $x_h=-0.375$.
    The horizontal line marks a one-percent relative error. This comparison tests relaxed endpoints,
    not the finite-time pole-to-pole traversal.}
    \label{fig:largeN-validation}
\end{figure}

\begin{table}[tbp]
    \centering
    \caption{Representative errors of the large-$N$ diagnostics. The Holstein--Primakoff column
    reports the relative error in the relaxed endpoint work obtained from separate pole-local
    approximations at $x_c=1$ and $x_h=-0.375$. The finite-time Kramers--Moyal columns report the
    maximum absolute errors in the rescaled corner mean $z=m/j$ and variance
    $V=\operatorname{Var}(z)$ over the two stationary-cycle corners. The Kramers--Moyal benchmark
    additionally uses $\gamma=0.1$ and $\tau_h=\tau_c=1$ as a reduced-model approximation
    diagnostic; the Davies-controlled production benchmark is the rate-scaled finite-contact
    comparison described in Sec.~\ref{sec:finite-numerics}.}
    \label{tab:largeN-errors}
    \begin{tabular}{@{}r r r r@{}}
    \toprule
    $N$ & HP work (rel.) & KM mean (abs.) & KM variance (abs.) \\
    \midrule
    8  & $5.78\times10^{-2}$ & $1.24\times10^{-1}$ & $3.92\times10^{-2}$ \\
    16 & $2.19\times10^{-3}$ & $7.68\times10^{-2}$ & $4.99\times10^{-2}$ \\
    32 & $4.77\times10^{-6}$ & $2.31\times10^{-2}$ & $1.04\times10^{-2}$ \\
    64 & $2.75\times10^{-11}$ & $9.02\times10^{-3}$ & $1.28\times10^{-2}$ \\
    \bottomrule
    \end{tabular}
\end{table}

\subsection{Symmetry-breaking local channels and Markov validity}
\label{sec:validity}

The ideal collective model assumes that the environment cannot distinguish which constituent
underwent a transition. Local noise breaks this indistinguishability: a jump acting on one
constituent can transfer population to other total-spin sectors, so it is not simply a correction
within the same reduced model unless the leakage is quantified. The relevant experimental question
is therefore how much symmetric-sector population is lost during the observation time. Let
$\Pi_{\rm sym}$ project onto the symmetric sector. For a Dicke state and local jump operators
$L_i$, the instantaneous loss of symmetric-sector weight is
\begin{equation}
 \left.\frac{\dd}{\dd t}\Tr(\Pi_{\rm sym}\rho)\right|_{\rm local}
 =-\sum_i\gamma_i
 \left\|(\id-\Pi_{\rm sym})L_i\ket{j,m}\right\|^2.
 \label{eq:leak-general}
\end{equation}

Because $H=-\Omega J_z$, the local operator $\sigma_i^+$ increases $m$ and therefore lowers the
energy; with this convention,
\begin{equation}
\begin{aligned}
\Pi_{\rm sym}\sigma_i^+\Pi_{\rm sym}&=\frac{J_+}{N},\\
\Gamma_{\rm leak}^{\downarrow}(m)
&=\gamma_{\rm loc}^{\downarrow}\frac{(j-m)(j-m-1)}{N}.
\end{aligned}
\label{eq:leak-down}
\end{equation}

The corresponding local energy-raising channel gives
\begin{equation}
    \Gamma_{\rm leak}^{\uparrow}(m)
    =\gamma_{\rm loc}^{\uparrow}
    \frac{(j+m)(j+m-1)}{N}.
    \label{eq:leak-up}
\end{equation}

For a population distribution $p_m(t)$, a first-order survival estimate is controlled by
\begin{equation}
    \Lambda_{\rm leak}
    =\int_0^{\tau_{\rm obs}}\!\dd t\sum_m p_m(t)
    \left[\Gamma_{\rm leak}^{\downarrow}(m)+
    \Gamma_{\rm leak}^{\uparrow}(m)\right].
    \label{eq:leak-integrated}
\end{equation}

The ideal Dicke-sector prediction remains controlled when $\Lambda_{\rm leak}\ll1$ over the observation
window. The simpler comparison $\gamma_{\rm loc}\ll\gamma$ provides an initial scale estimate,
but the integrated criterion retains the rung and $N$ dependence of the leakage process.

There is a separate large-$N$ restriction on the Markovian Davies approximation. For a Dicke rung
$m$, define the total outgoing rate on contact $\alpha$ by
\begin{equation}
\begin{aligned}
\Gamma_{{\rm out},\alpha}(m)={}&
\Gamma_{\downarrow,\alpha}(j-m)(j+m+1)\\
&+\Gamma_{\uparrow,\alpha}(j+m)(j-m+1),\\
\Gamma_{\max}(N)={}&\max_{\alpha,m}\Gamma_{{\rm out},\alpha}(m).
\end{aligned}
\label{eq:gamma-max-def}
\end{equation}

Under the matched-total-rate convention, the discrete maximum for even $N$ occurs at $m=0$ and is
$\Gamma_{\max}(N)=\gamma N(N+2)/4$. For the finite-contact production benchmark we adopt the
conservative numerical criterion
\begin{equation}
\epsilon_{\rm sec}(N)\equiv
\frac{\Gamma_{\max}(N)}{\Omega_{\min}}\le0.05,
\qquad \Omega_{\min}=\min(\Omega_c,\Omega_h),
\label{eq:secular-control}
\end{equation}
where $0.05$ is a benchmark acceptance threshold rather than a universal Davies constant. Choosing
$\gamma=4\times10^{-5}$ gives $\epsilon_{\rm sec}(64)=0.04224$ and therefore satisfies this
criterion for every finite-contact size used up to $N=64$. The contact durations are increased at
fixed $u=\gamma\tau$, so the population propagators are unchanged while absolute powers rescale
with $\gamma$.

A microscopic Markovian description also requires a short bath correlation time. If $\tau_B$ is
that correlation time, the remaining weak-coupling/separation conditions are
\begin{equation}
    \Gamma_{\max}(N)\tau_B\ll1,
    \qquad
    \Gamma_{\max}(N)\ll\Omega_\alpha.
    \label{eq:markov-range}
\end{equation}

The first inequality cannot be assigned a numerical value without specifying a microscopic bath
correlation time. At fixed microscopic coupling, collective rates grow with $N$, so finite-time
comparisons must state both their rate-matching convention and the range over which the secular
criterion is controlled.

Equations~\eqref{eq:leak-general}--\eqref{eq:leak-integrated} give a rung-resolved sector-survival
diagnostic for the specified weak local channels. Loss of symmetric-sector population during a
contact directly measures the breakdown of the ideal collective description.

Trapped-ion platforms offer independent control over coherent spin operations and engineered
dissipation~\cite{Barreiro2011OpenSystem}, including thermal reservoirs with tunable temperatures
and dissipation rates~\cite{So2026ThermalReservoirs}. Ultracold atoms coupled to optical cavities
provide a complementary route to collective light--matter coupling, inversion, and superradiant
readout~\cite{Bohr2024SuperradiantReadout}, while bounded cold-atom spectra support
negative-temperature states~\cite{Braun2013NegativeTemperature}. In either setting, the quantities
needed for the present tests are operational: the gap protocol fixes $\Omega_c$ and $\Omega_h$,
repeated measurements of the collective polarization reconstruct the corner distributions
$p_1(m)$ and $p_3(m)$, and cycle-resolved records determine the mean work, work variance, and
inter-cycle correlations. Independent characterization of symmetry leakage and the engineered
reservoir rates would then test the two principal validity conditions of the collective model.

\section{Conclusion and Outlook}
\label{sec:conclusion}

This work separates three effects that can easily be conflated in a collective population-inverted
heat engine: the thermodynamic effect of inversion, the kinetic effect of collective coupling, and
the energetic resource already stored in the inverted state. The finite Dicke ladder makes this
separation particularly transparent.

In the complete-reset limit, the main distinction between passive and inverted operation is
geometric. For fixed reservoir conditions away from infinite temperature, the cold and passive-hot
populations remain localized within $O(1)$ rungs of the same end of the Dicke ladder as $N$
increases. Their mean polarizations therefore remain separated by only $O(1)$, and the extracted
work per cycle saturates. Population inversion moves the hot distribution to the opposite end of
the ladder while leaving its width finite. The cold and hot populations are then separated by
$O(N)$ rungs, so the work grows linearly with the number of constituents. The same picture explains
the fluctuation results: matched passive and inverted hot states have the same local width and the
same even hot-state fluctuation contributions, but inversion turns an $O(1)$ separation of
distribution centers into an $O(N)$ one. Consequently, the work reliability remains $O(1)$ on the
passive branch but grows as $O(N)$ on the inverted branch, exceeding the $O(\sqrt N)$
self-averaging obtained from $N$ independent engines.

This extensive complete-reset work should not be confused with a collective acceleration of the
dynamics. It is an endpoint effect produced by the macroscopic separation of the stationary
populations. Collective coupling plays a different role: the Dicke matrix elements reshape the
transition rates and can accelerate relaxation at finite contact times. Near a pole this
enhancement is captured by the local Holstein--Primakoff description, but an inverted cycle must
carry population across an $O(N)$ distance between opposite ends of the ladder and therefore lies
outside any single pole-local approximation. The exact finite-state dynamics nevertheless shows
that, under the matched-total-rate convention used here, collective kinetics can increase the gross
inverted power relative to $N$ independent engines. Finite contacts also preserve part of the
incoming polarization, generating correlations between cycle corners and, over repeated operation,
between successive work outputs. A stationary mean output therefore does not imply statistically
independent cycles: one-cycle reliability and long-time stability can differ whenever the contacts
do not fully erase the boundary-state memory.

The resource analysis gives a complementary interpretation of the enhanced inverted output. For
matched passive and inverted hot states, the additional gross work is exactly
\begin{equation}
    \Delta W_{\rm gross}=\eta W_{\rm erg}.
\end{equation}

Thus the extra Otto work is obtained by converting a fraction $\eta$ of the ergotropy already
carried by the inverted state. If the reversible excess state-formation cost is charged once per
cycle, the incremental balance relative to the matched passive engine becomes
\begin{equation}
    \Delta W_{\rm net}=-(1-\eta)W_{\rm erg}<0.
\end{equation}

The larger gross output of the inverted engine is therefore not a free thermodynamic gain: it
converts a pre-existing active-state resource. The same conclusion holds for the wider class of
passive comparators considered here under the complete-reset, cold-reference reversible
state-level accounting. Population inversion can nevertheless remain operationally valuable when
the active resource is externally supplied, persistent over several cycles, or when power and
reliability rather than net energetic gain are the relevant performance criteria.

The resulting physical hierarchy is therefore simple. Population inversion determines how far
apart the thermodynamic endpoints lie on the finite ladder. Collective coupling determines how
quickly the system can move through that ladder at finite time. Resource accounting determines
whether the resulting increase in gross work represents a genuine net advantage or the conversion
of energy already invested in preparing an active state. The exact finite-state treatment, together
with the Holstein--Primakoff, Kramers--Moyal, symmetry-leakage, and secular-control diagnostics,
shows where each of these statements remains quantitatively controlled as the system size and
contact times are varied.

Several extensions follow naturally from this separation of mechanisms. Noncommuting work strokes
would introduce coherence and make the work statistics genuinely quantum beyond the present
Dicke-diagonal description. Finite reservoir depletion would promote the active resource from a
fixed boundary condition to a dynamical variable, while pump-resolved or autonomous reservoir
models would connect the state-level formation bounds used here to an explicit device-level energy
balance. Repeated-interaction and autonomous active reservoirs
~\cite{strasberg2017repeated,niedenzu2019autonomous,elouard2023autonomous,
shaghaghi2022random}, as well as squeezed-reservoir engines
~\cite{manzano2018squeezed,klaers2017squeezed,xiao2023squeezed}, provide natural settings in which
to test how collective kinetics, work fluctuations, and resource costs change when the reservoir
itself becomes part of the thermodynamic system.

\begin{acknowledgments}
This study was financed in part by the Coordena\c{c}\~ao de Aperfei\c{c}oamento de Pessoal de N\'ivel Superior -- Brasil (CAPES) -- Finance Code 001. We are grateful to Rogério J. de Assis for insightful discussions and suggestions. G.G.D thanks Professor Gao Xianlong for his kind hospitality and support during their stay at the Zhejiang Normal University. Generative AI tools were used for language editing and literature-search assistance; all scientific content was reviewed and validated by the authors, who take full responsibility for the manuscript.
\end{acknowledgments}

\section*{Data and code availability}
The data and numerical code supporting the findings of this study are available from the
corresponding author upon reasonable request.

\appendix

\section{Collective-spin algebra and sector preservation}
\label{app:collective-spin}

The collective-spin algebra used in the main text follows from
$J_\nu=\frac12\sum_{r=1}^N\sigma_\nu^{(r)}$ and
$J_\pm=\sum_r\sigma_\pm^{(r)}$. Operators acting on different constituents commute, while
$[\sigma_z^{(r)},\sigma_\pm^{(r)}]=\pm2\sigma_\pm^{(r)}$. Therefore
$[J_z,J_\pm]=\pm J_\pm$ and $[J_+,J_-]=2J_z$. The Casimir operator $J^2=J_x^2+J_y^2+J_z^2$ commutes with all three collective generators.
Its simultaneous eigenstates are denoted by $|j,m\rangle$, with
$J^2|j,m\rangle=j(j+1)|j,m\rangle$ and $J_z|j,m\rangle=m|j,m\rangle$.

Using $J_\mp J_\pm=J^2-J_z^2\mp J_z$, the norm is
$\|J_\pm|j,m\rangle\|^2=(j\mp m)(j\pm m+1)$, which fixes the phase convention used in Eq.~\eqref{eq:ladder}.

The full tensor-product Hilbert space decomposes into sectors with different $j$. The fully
permutation-symmetric sector has $j=N/2$ and dimension $2j+1=N+1$. If $P_j$ denotes its
projector, then $[P_j,J_z]=[P_j,J_\pm]=0$. Consequently, a Hamiltonian and jump operators
constructed only from these collective operators cannot transfer population between different
$j$ sectors. This is the precise assumption behind the reduction used here. A local
symmetry-breaking jump operator would generally invalidate it.

\section{Birth--death reduction}
\label{app:birth-death}

Let $\rho=\sum_m p(m)|j,m\rangle\langle j,m|$ be diagonal. The jump $J_+$ maps $m$ to
$m+1$ with squared matrix element $(j-m)(j+m+1)$, while $J_-$ maps $m$ to $m-1$ with squared
matrix element $(j+m)(j-m+1)$. Substitution into the dissipator gives
\begin{equation}
\begin{split}
    \dot p(m)={}&w_\alpha^+(m-1)p(m-1)+w_\alpha^-(m+1)p(m+1)\\
    &-[w_\alpha^+(m)+w_\alpha^-(m)]p(m),
\end{split}
\label{eq:app-bd}
\end{equation}
where
\begin{equation}
\begin{aligned}
w_\alpha^+(m)&=\Gamma_{\downarrow,\alpha}(j-m)(j+m+1),\\
w_\alpha^-(m)&=\Gamma_{\uparrow,\alpha}(j+m)(j-m+1).
\end{aligned}
\label{eq:app-rates}
\end{equation}

The superscripts indicate motion toward larger or smaller $m$, respectively, and not the
population bias by themselves.

Writing $\dot{\vec p}=R_\alpha\vec p$, the columns of $R_\alpha$ sum to zero. The matrix
$K_\alpha(\tau_\alpha)=\exp(R_\alpha\tau_\alpha)$ is therefore a stochastic propagator: its entries
are nonnegative and each column sums to one. The stationary cycle is the normalized solution of
$\vec p_1=K_cK_h\vec p_1$. For positive elementary rates, all neighboring rungs communicate and the fixed point is unique.

For an infinitely long contact, detailed balance requires
$w_\alpha^+(m)p(m)=w_\alpha^-(m+1)p(m+1)$. Since the common Dicke factor cancels,
$p(m+1)/p(m)=\Gamma_{\downarrow,\alpha}/\Gamma_{\uparrow,\alpha}=e^{x_\alpha}$, which gives the Gibbs-Dicke distribution in Eq.~\eqref{eq:gibbs-dicke}.

\section{One-cycle trajectory statistics}
\label{app:trajectory}

The positive trajectory description follows from the population process. If $m_1$ is sampled from
the stationary corner distribution, the hot transition to $m_3$ has probability
$[K_h]_{m_3,m_1}$ and the cold transition to $m_5$ has probability
$[K_c]_{m_5,m_3}$. Thus
\begin{equation}
    {\cal P}(m_1,m_3,m_5)
    =p_1(m_1)[K_h]_{m_3,m_1}[K_c]_{m_5,m_3}.
    \label{eq:app-path}
\end{equation}

Normalization follows from the normalization of $p_1$ and the column-stochasticity of both
propagators.

The corner energies are
$E_c(m_1)=-\Omega_cm_1$, $E_h(m_1)=-\Omega_hm_1$,
$E_h(m_3)=-\Omega_hm_3$, $E_c(m_3)=-\Omega_cm_3$, and
$E_c(m_5)=-\Omega_cm_5$. Therefore
\begin{align}
    W_{\rm comp}&=-\Delta\Omega\,m_1,&
    W_{\rm exp}&=+\Delta\Omega\,m_3,\\
    Q_h&=-\Omega_h(m_3-m_1),&
    Q_c&=\Omega_c(m_3-m_5).
\end{align}

The fifth-corner label is needed to assign the cold-contact heat and to test cycle closure, but it
does not enter the work because the cold contact occurs after the expansion. Accordingly,
$W+Q_h+Q_c=E_c(m_5)-E_c(m_1)$ on an individual path, while the stationary average obeys the
cycle first law because the distributions of $m_5$ and $m_1$ coincide.

For any integer $r\geq0$, the raw work moment is
\begin{equation}
    \langle W^r\rangle
    =\sum_{m_1,m_3,m_5}{\cal P}(m_1,m_3,m_5)
    [\Delta\Omega(m_3-m_1)]^r.
    \label{eq:app-work-moment}
\end{equation}

The first two cumulants follow from
$\operatorname{Var}(W)=\langle W^2\rangle-\langle W\rangle^2$. Mixed work--heat moments are
obtained by replacing the final factor with the corresponding product of the expressions above.

The complete-reset limit replaces each conditional transition column by the stationary
distribution. Equation~\eqref{eq:app-path} then factorizes at the hot endpoints, and summing over
$m_5$ recovers Eq.~\eqref{eq:factorization}. For a finite numerical sum, omitting small
probability entries is an approximation. A calculation that uses such a cutoff must repeat the
sum with decreasing cutoffs and report the change in normalized moments.

\section{Free-energy identities and resource references}
\label{app:resource}

For $\rho_x=e^{xJ_z}/Z_j(x)$, one has
$\ln\rho_x=xJ_z-\ln Z_j(x)$ and therefore
$S_N(x)=-\Tr(\rho_x\ln\rho_x)=\ln Z_j(x)-x\mu_N(x)$. Differentiating and using
$\partial_x\mu_N(x)=\langle J_z^2\rangle_x-\langle J_z\rangle_x^2$ gives
$dS_N/dx=-x\,\partial_x\mu_N(x)$.

Let $\rho_{x_0}$ be the Gibbs-Dicke equilibrium state for the Hamiltonian
$H_h=-\Omega_hJ_z$ at reference temperature $T=\Omega_h/x_0$. Then
\begin{align}
    T D(\rho_x\|\rho_{x_0})
    &=T\Tr[\rho_x(\ln\rho_x-\ln\rho_{x_0})]\\
    &=F(x;\Omega_h,T)-F(x_0;\Omega_h,T).
\end{align}

This identity is a lower bound on reversible state-formation work for the selected reference. It
does not specify the work required by a nonequilibrium protocol with finite time, dissipation, or
control constraints.

For $a>0$, the probability vectors satisfy
$p_{-a}(-m)=p_{+a}(m)$. The energies satisfy
$E_h(-m)-E_h(m)=2\Omega_hm$, so the passive rearrangement of $\rho_{-a}$ is $\rho_{+a}$. The
entropy is unchanged by rearrangement and the energy difference is
\begin{equation}
    F(-a;\Omega_h,T)-F(+a;\Omega_h,T)
    =2\Omega_h\mu_N(a)=W_{\rm erg}(a).
    \label{eq:app-ergotropy}
\end{equation}

The product-reference expression and the correlation decomposition are given in
Eqs.~\eqref{eq:tau-product} and \eqref{eq:corr-decomp}. The two references give the same excess
difference between $\rho_{-a}$ and $\rho_{+a}$, but their total formation bounds need not be equal.
A total formation cost is therefore defined only after its reference is specified.

\section{Large-\texorpdfstring{$N$}{N} expansions}
\label{app:largeN}

Starting from $Z_j(x)=\sinh[(N+1)x/2]/\sinh(x/2)$ and using
$\coth y=\operatorname{sgn}(y)+O(e^{-2|y|})$ gives, for fixed nonzero $x$,
\begin{align}
    \mu_N(x)&=\frac N2-\frac{1}{e^x-1}
    +O(Ne^{-Nx}),&&x>0,\\
    \mu_N(x)&=-\frac N2+\frac{1}{e^{-x}-1}
    +O(Ne^{-N|x|}),&&x<0.
\end{align}

The passive and inverted complete-reset work scalings follow by inserting these two expressions
into Eq.~\eqref{eq:work-complete-reset}.

For the north-pole coordinate $k=j-m$, the exact transition rates are
\begin{equation}
\begin{aligned}
w_\alpha(k\to k-1)&=\Gamma_{\downarrow,\alpha}k(N-k+1),\\
w_\alpha(k\to k+1)&=\Gamma_{\uparrow,\alpha}(N-k)(k+1).
\end{aligned}
\end{equation}

When $k=O(1)$ and $N\to\infty$, the leading drift of $k$ is
\begin{equation}
    \frac{d\langle k\rangle}{dt}
    =N\Gamma_{\uparrow,\alpha}
    -N(\Gamma_{\downarrow,\alpha}-\Gamma_{\uparrow,\alpha})\langle k\rangle.
\end{equation}

Under the matched-total-rate convention,
$\Gamma_{\downarrow,\alpha}-\Gamma_{\uparrow,\alpha}
=\gamma\tanh(x_\alpha/2)$ and
$\Gamma_{\uparrow,\alpha}/[\gamma\tanh(x_\alpha/2)]=1/(e^{x_\alpha}-1)$,
which gives Eqs.~\eqref{eq:k-ode} and \eqref{eq:k-sol}.

This expansion requires both a pole-localized distribution and positive $x_\alpha$ for the north
pole to be stable. An inverted contact drives the system toward the opposite pole; the local
equation applies only while the distribution remains localized near the stable pole and does not
describe the $O(N)$ traversal between poles.

\section{Numerical definitions and validation criteria}
\label{app:numerical}

All numerical quantities are evaluated in the dimensionless units $\hbar=k_{\rm B}=1$ used in the
main text. The numerical energy unit is chosen to be unity; changing this reference scale rescales
all energies, work, heat, and power consistently. In the numerical implementation used for the
reported finite-state calculations, the tolerances for probability negativity and stationary-cycle
convergence are $10^{-13}$ in the population representation, and trajectory probabilities below
$10^{-16}$ may be omitted from stored path tables after renormalization. The Kramers--Moyal linear-noise equations are integrated with relative
and absolute ODE tolerances $10^{-10}$ and $10^{-12}$, respectively, and their stationary-cycle
iteration uses a $10^{-12}$ convergence threshold.

For each $N$, the population generator is constructed from Eq.~\eqref{eq:app-rates}. The contact
maps are computed as matrix exponentials. The stationary corner is obtained by iterating
$K_cK_h$ until the $\ell^1$ difference between successive population vectors is below the stated
tolerance. We verify the following conditions:
\begin{enumerate}
    \item every population vector is nonnegative and normalized;
    \item every propagator is nonnegative and column-stochastic;
    \item the stationary-cycle residual
    $\lVert K_cK_h\vec p_1-\vec p_1\rVert_1$ is below tolerance;
    \item the mean first law in Eq.~\eqref{eq:first-law-cycle} closes;
    \item the path first law in Eq.~\eqref{eq:path-first-law} closes for every retained path;
    \item complete-reset results converge to Eqs.~\eqref{eq:means-complete-reset},
    \eqref{eq:var-W}, and \eqref{eq:RW-def};
    \item the relative-entropy cost identity in Eq.~\eqref{eq:free-energy-relative-entropy}
    closes to numerical precision.
\end{enumerate}

Finite-contact trajectory sums are deterministic rather than Monte Carlo estimates. If probabilities
below a numerical threshold are omitted to reduce storage, the calculation is repeated with
successively smaller thresholds. The work and heat moments are reported only when their changes are smaller than the prescribed
numerical tolerance. The finite-state calculation is then
exact up to matrix-exponential, stationary-iteration, floating-point, and any explicitly reported
path-truncation errors.

The large-$N$ comparisons use the exact finite-state result as the reference. Holstein--Primakoff
errors are defined from the relative difference of the relaxed work obtained by combining separate
pole-local approximations at the two relaxed endpoints, while Kramers--Moyal errors are defined
separately for the first and second moments. Agreement of a mean does not validate the variance,
and agreement at the relaxed endpoints does not validate a finite-time traversal between opposite
poles.

The numerical results cover the parameter sets specified in the figures and tables. Asymptotic
statements are restricted to the regimes derived analytically, while finite-contact quantities are
reported as finite-$N$ evaluations of the exact population model.

\bibliography{references}

\end{document}